%% file: main_arxiv.tex
\documentclass{article}

 \usepackage[preprint]{neurips_2026}

\usepackage[utf8]{inputenc} 
\usepackage[T1]{fontenc}    
\usepackage{hyperref}       
\usepackage{url}            
\usepackage{booktabs}       
\usepackage{amsfonts}       
\usepackage{nicefrac}       
\usepackage{microtype}      
\usepackage{xcolor}         
\usepackage{subcaption}

\usepackage{amsmath}
\usepackage{amsthm}
\usepackage{amssymb}
\usepackage{thmtools}
\usepackage{cleveref}
\usepackage{algorithm}
\usepackage{comment}
\usepackage[noend]{algpseudocode} 
\usepackage{tikz}
\usepackage{pgfplots}
\usepgfplotslibrary{groupplots}
\pgfplotsset{compat=1.18}

\title{Optimal Learning-Augmented Algorithm\\for Online Bidding}

\author{%
  Changyeol Lee\footnotemark[1]\\
  Dept.\ of Computer Science and Engineering\\
  Yonsei University\\
  Seoul, South Korea \\
  \texttt{777john@yonsei.ac.kr} \\
  \And
  Dahoon Lee\footnotemark[1]\\
  Dept.\ of Mathematical Sciences \\
  Seoul National University \\
  Seoul, South Korea \\
  \texttt{dahoon46@snu.ac.kr} \\
  \AND
  Jongseo Lee\footnotemark[1]\\
  School of Computing \\
  KAIST \\
  Daejeon, South Korea\\
  \texttt{leejseo@kaist.ac.kr} \\
  \And
  Yongho Shin\footnotemark[1]\\
  Institute of Computer Science \\
  University of Wroc{\l}aw\\
  Wroc{\l}aw, Poland \\
  \texttt{yongho@cs.uni.wroc.pl} \\
  \And
  Changki Yun\thanks{Equal contribution}\\
  Dept.\ of Mathematical Sciences \\
  Seoul National University \\
  Seoul, South Korea \\
  \texttt{tamref.yun@snu.ac.kr} \\
}

\theoremstyle{plain}
\newtheorem{theorem}{Theorem}

\newtheorem{lemma}[theorem]{Lemma}
\newtheorem{corollary}[theorem]{Corollary}

\newtheorem{fact}[theorem]{Fact}

\theoremstyle{definition}
\newtheorem{definition}[theorem]{Definition}

\theoremstyle{remark}
\newtheorem{remark}[theorem]{Remark}

\input{macro}

\begin{document}
\maketitle
\begin{abstract}
    Recent advances in machine learning have spurred significant interest in learning-augmented algorithms, particularly for online optimization.
    A growing body of work has studied \emph{online bidding} in this framework, aiming to characterize the trade-off between robustness and consistency.
    While this trade-off is fully understood for deterministic algorithms, a gap between upper and lower bounds remains in the randomized setting.
    In this paper, we close this gap by presenting a Pareto-optimal randomized learning-augmented algorithm for this problem.
    Our approach introduces the notion of a \emph{bidding profile}, a novel framework for representing the distribution over bids generated by an algorithm.
    We show that any bidding algorithm can be reduced, without loss of generality, to one driven by a bidding profile, and we characterize the optimal profile via a system of delayed differential equations.
    Finally, we demonstrate the broader applicability of our approach by extending it to the \emph{linear search} problem, yielding a significant improvement over prior learning-augmented algorithms for linear search.
\end{abstract}

\input{arxiv_include/1_introduction}
\input{arxiv_include/2_preliminaries}
\input{arxiv_include/3_optimal_bidding}
\input{arxiv_include/4_linear_search}
\input{arxiv_include/5_conclusion}

\begin{ack}
CL: This work was partly supported by an IITP grant funded by the Korean Government (MSIT) (No. RS-2020-II201361, Artificial Intelligence Graduate School Program (Yonsei University)). 
DL and CY: The work of Dahoon Lee and Changki Yun was supported in part by the National Research Foundation of Korea (NRF) under Grant No. RS-2026-25484051 and by the New Faculty Startup Fund of Seoul National University.
YS: This work was supported by the NCN grant no.\ 2020/39/B/ST6/01641.
The authors acknowledge the use of LLMs for assisting with some proofs, editing the manuscript, and generating figures.  
All LLM-generated content has been carefully verified by the authors.
\end{ack}

\bibliographystyle{alpha}
\bibliography{refs}

\appendix
\input{arxiv_include/A_relwork}
\input{arxiv_include/B_moski_ob}
\input{arxiv_include/C_bid}
\input{arxiv_include/D_lin}
\input{arxiv_include/E_str2alg}

\end{document}

%% file: macro.tex
\newcommand{\cost}{\mathsf{cost}}

\newcommand{\roblin}{\rho_{\mathrm{LS}}}
\newcommand{\conslin}{\chi_{\mathrm{LS}}}

\newcommand{\scriptA}{\mathcal{A}}
\newcommand{\bbE}[1]{\mathbb{E}\left[#1 \right]}
\newcommand{\bbR}{\mathbb{R}}
\newcommand{\bbZ}{\mathbb{Z}}
\newcommand{\dd}{\mathrm{d}}

\newcommand{\cA}{\mathcal{A}}
\newcommand{\cB}{\mathcal{B}}
\newcommand{\cF}{\mathcal{F}}

\newcommand{\Unif}{\mathrm{Unif}}

\newcommand{\dir}{\mathrm{dir}}

\DeclareMathOperator{\sgn}{sgn}

%% file: arxiv_include/1_introduction.tex
\section{Introduction}
Recent advances in machine learning (ML) have spurred active research on \emph{learning-augmented algorithms}~\cite{mitzenmacher2022algorithms}.  
These algorithms take predictions as additional inputs and leverage them to solve computational problems.  
A key advantage of this paradigm is that it exploits the empirically high quality of ML predictions to achieve strong performance, while still providing robust guarantees regardless of prediction accuracy.  
Another distinguishing feature is that the prediction model is treated as a complete black box, making the framework independent of any specific prediction mechanism.

Research on learning-augmented algorithms has been particularly active in the context of \emph{online optimization}, where predictions can naturally serve as surrogates for unknown future inputs.  
Since the seminal work of Lykouris and Vassilvitskii~\cite{lykouris2021competitive}, a substantial body of work, including~\cite{purohit2018improving,bamas2020primal,azar2022online,elenter2024overcoming,shin25,choo2025learning}, has studied a variety of online optimization problems in the learning-augmented setting, demonstrating the effectiveness of this paradigm.

The \emph{online bidding} problem is a canonical online optimization problem in which an algorithm sequentially submits bids until reaching an unknown target value $T > 0$, with the objective of minimizing the total sum of submitted bids.  
This problem succinctly captures the core challenge of searching for an unknown target in online optimization.  
For example, it models the design of \emph{anytime algorithms} or \emph{interruptible algorithms}---which arise frequently in intelligent systems and real-time applications~\cite{zilberstein1996using,karaman2011anytime}---through the scheduling of \emph{contract algorithms}~\cite{russell1991composing,lopez2014optimal}.\footnote{%
A contract algorithm receives its execution time as part of the input and provides a performance guarantee only if allowed to run for the prescribed duration, whereas an anytime (or interruptible) algorithm can be interrupted at any time and return its best current solution.}  
In this setting, one schedules a sequence of contract algorithm executions with prescribed running times so as to achieve the best possible performance guarantee under an unknown interruption time $T$.  
The resulting trade-off between solution quality and execution time is precisely captured by online bidding.

The standard performance measure for online bidding algorithms is the \emph{competitive ratio}, defined as the worst-case ratio between the (expected) total sum of submitted bids and the optimal bid with hindsight (i.e., the target value $T$).  
A classical folklore algorithm for this problem is the \emph{doubling method}, which doubles the bid in each round until termination.  
Chrobak, Kenyon, Noga, and Young~\cite{chrobak2008incremental} showed that this algorithm is 4-competitive and that no deterministic algorithm can achieve a better competitive ratio. For randomized algorithms, one can randomly shift the initial bid to have $e^U$ with $U \sim \Unif(0,1]$, and then increase the bid by a factor of $e$ in every round.
This yields an $e$-competitive algorithm, which was also shown to be optimal among all randomized algorithms~\cite{chrobak2008incremental}.

Following the growing interest in learning-augmented algorithms, online bidding has been extensively studied in this framework, particularly in the setting where a single value prediction $P$ is provided as an estimate of the unknown target value $T$~\cite{anand2021regression,im2023online,angelopoulos2024online,shin25,angelopoulos2025learning}.  
A central focus of this line of research is the trade-off between \emph{robustness} and \emph{consistency}~\cite{lykouris2021competitive,purohit2018improving}.  
Informally, robustness measures the competitive ratio guaranteed regardless of the prediction quality, whereas consistency captures the performance when the prediction is accurate, i.e., when $P = T$.

One might expect that a Pareto-optimal learning-augmented algorithm can be obtained by adapting an optimal competitive algorithm without learning augmentation.  
This is indeed the case for deterministic online bidding: Angelopoulos, D{\"u}rr, Jin, Kamali, and Renault~\cite{angelopoulos2024online} and Shin, Lee, Lee, and An~\cite{shin25}\footnote{Shin et al.~\cite{shin25} studied multi-option ski rental, but their techniques extend naturally to online bidding. See Appendix~\ref{app:moskiob} for the equivalence between the two problems.} independently showed that a simple modification of the doubling method yields a deterministic Pareto-optimal learning-augmented algorithm.  
In particular, the algorithm achieves strong consistency by deterministically including the prediction $P$ into its sequence of geometrically increasing bids.

In contrast, a straightforward extension is less apparent for the randomized $e$-competitive algorithm, since the random shift in the bid sequence makes it difficult to deterministically incorporate the prediction $P$.  
To address this issue, Shin et al.~\cite{shin25} introduced a ``prediction-following'' phase, in which the algorithm explicitly submits the bid $P$ during the execution of the randomized algorithm, yielding the currently best known trade-off.  
Alternatively, Angelopoulos and Simon~\cite{angelopoulos2025learning} carefully interpolated between the optimal randomized competitive algorithm and the optimal deterministic learning-augmented algorithm through controlled randomization.  
Despite these advances, a gap remains between the best-known upper and lower bounds on the robustness-consistency trade-off~\cite{shin25}, leaving the characterization of the Pareto-optimal trade-off as an open question.  
See Figure~\ref{subfig:consistency-robustness-tradeoff-bidding} for a comparison of the known trade-offs.

\input{arxiv_include/1a_introduction_contribution}

%% file: arxiv_include/1a_introduction_contribution.tex
\input{arxiv_include/1b_introduction_tikz}
\subsection{Our contributions}
In this paper, we resolve this open question by presenting a randomized learning-augmented algorithm that achieves the Pareto-optimal robustness-consistency trade-off (see \Cref{subfig:consistency-robustness-tradeoff-bidding}).

\begin{theorem} \label{thm:main-bidding}
    For any trade-off parameter $s \in (0, 1]$, there exists a randomized $\rho(s)$-robust $\chi(s)$-consistent online bidding algorithm with
    \[
    \rho(s):=\frac{e^s}{s}
    \quad\text{ and }\quad
    \chi(s):=\begin{cases}
    (e^s-1)/s, & 0< s < \ln 2, \\
    \xi(s)/s, & \ln 2 \leq s \leq 1,
    \end{cases}
    \]
    where $\xi(s)$ is the unique solution $ \xi\in[1,e] $ of $\xi(2-\ln \xi)=e^s$ for $s \in [\ln 2, 1]$.
    Moreover, the robustness-consistency trade-off of this algorithm is Pareto-optimal.
\end{theorem}

To achieve this goal, we introduce a novel framework, called a \emph{bidding profile}, for representing the probability distribution over the bids generated by an algorithm.  
Intuitively, a bidding profile is specified by a non-decreasing function $G : \bbR \to (0, \infty)$, from which the bid sequence is generated through a random shift.  To gain intuition for the role of a bidding profile, consider the classical randomized $e$-competitive algorithm. Recall that the algorithm chooses the initial bid as $e^U$ with $U \sim \Unif(0,1]$, and subsequently multiplies the bid by a factor of $e$ in every round until termination. This execution can be naturally associated with the exponential function $G(t) := e^t$, in the sense that the bids generated by the algorithm form a subsequence of the infinite random sequence $\{G(n+U)\}_{n \in \bbZ}$, where $U \sim \Unif(0,1]$.

Motivated by this perspective, we study the family of online bidding algorithms induced by bidding profiles $G$ and identify the profile whose corresponding algorithm achieves the optimal robustness-consistency trade-off within this family via a tailored system of delayed differential equations.  
We then show that any online bidding algorithm can be transformed into one induced by a bidding profile without loss in performance guarantees, thereby establishing the Pareto-optimality of our bidding-profile-based algorithm.

We further demonstrate the versatility of our profile-based approach by extending it to the \emph{linear search} problem, yielding a significant improvement in the robustness--consistency trade-off for randomized learning-augmented linear search algorithms over the previous state of the art~\cite{angelopoulos2025learning}.  
However, our lower-bound argument for online bidding does not extend naturally to the linear search setting, and thus establishing the Pareto-optimality of the resulting algorithm remains an open problem.  
Instead, we derive a nearly tight lower bound using Yao's principle, building on ideas from prior work~\cite{angelopoulos2025learning,shin25}.  
See~\Cref{subfig:linear-search-consistency-robustness-search} for a comparison of our upper and lower bounds.

During the preparation of this paper, we became aware of concurrent and independent work by Degryse, Zaakour, D\"urr, and Angelopoulos~\cite{degryse2026} on learning-augmented online bidding.  
They also establish the Pareto-optimal robustness-consistency trade-off for robustness ratios at least $2.885$, covering almost the entire range of robustness.
In addition, they extend their results to the incremental median problem.

Further related work is discussed in Appendix~\ref{app:relwork}.

%% file: arxiv_include/1b_introduction_tikz.tex
\pgfplotsset{set layers}

\begin{figure}[t]
\centering
\begin{subfigure}{0.49\textwidth}
\centering
\begin{tikzpicture}
\begin{axis}[
    width=\textwidth,
    xlabel={Consistency $\chi$},
    ylabel={Robustness $\rho$},
    xmin=1, xmax=2.8,
    xmode=log,
    log basis x={2.718281828},
    xtick={1, 1.5, 2, exp(1)},
    xticklabels={$1$, $1.5$, $2$, $e$},
    ymin=2.4, ymax=30,
    ymode=log,
    log basis y={2.718281828},
    ytick={exp(1), 5, 10, 20},
    yticklabels={$e$, $5$, $10$, $20$},
    grid=both,
    legend pos=north east,
    samples=300,
    domain=0.035:0.69314718056,
    font=\footnotesize,
    set layers
]

\addplot[
    on layer=axis foreground,
    thick,
    blue,
    variable=\s,
]
(
    {(exp(\s)-1)/\s},
    {exp(\s)/\s}
);
\addlegendentry{Our alg \& LB}

\addplot[
    on layer=axis foreground,
    thick,
    blue,
    variable=\y,
    forget plot,
    domain=1:2.71828182846,
]
(
    {\y / ln(\y*(2-ln(\y)))},
    {\y*(2-ln(\y)) / ln(\y*(2-ln(\y)))}
);

\addplot[
thick,dashed,orange,
variable=\t,
domain=0.01:1,
]
(
{1+\t},
{max(2.71828182846,(1+\t)^2/(2*\t))}
);
\addlegendentry{LB~\cite{shin25}}


\addplot[
    thick,dashdotted,red,
    variable=\t,
    domain=0.01:0.08612677965,
]
(
    {1+\t},
    {2.71828182846*(1+\t)^2/(4*\t)}
);
\addlegendentry{Alg~\cite{shin25}}
\addplot[
    thick,dashdotted,red,
    smooth,
    forget plot,
]
coordinates {
    (1.020000,35.351255)
    (1.025000,28.558948)
    (1.030000,24.031877)
    (1.035000,20.799225)
    (1.040000,18.375585)
    (1.050000,14.984529)
    (1.060000,12.726089)
    (1.070000,11.114860)
    (1.080000,9.908137)
    (1.086127,9.308020)
    (1.113329,7.470894)
    (1.140532,6.389804)
    (1.167735,5.680504)
    (1.194937,5.180589)
    (1.222140,4.809850)
    (1.249342,4.524206)
    (1.276545,4.297470)
    (1.303747,4.113135)
    (1.330950,3.960298)
    (1.358153,3.831483)
    (1.385355,3.721398)
    (1.412558,3.626203)
    (1.439760,3.543042)
    (1.466963,3.469755)
    (1.494166,3.404675)
    (1.521368,3.346498)
    (1.548571,3.294191)
    (1.575773,3.246920)
    (1.602976,3.204010)
    (1.630178,3.164906)
    (1.657381,3.129147)
    (1.684584,3.096348)
    (1.711786,3.066185)
    (1.738989,3.038381)
    (1.766191,3.012700)
    (1.793394,2.988938)
    (1.820597,2.966920)
    (1.847799,2.946492)
    (1.875002,2.927519)
    (1.902204,2.909883)
    (1.929407,2.893479)
    (1.956609,2.878215)
    (1.983812,2.864008)
    (2.011015,2.850784)
    (2.038217,2.838477)
    (2.065420,2.827026)
    (2.092622,2.816377)
    (2.119825,2.806482)
    (2.147028,2.797297)
    (2.174230,2.788779)
    (2.201433,2.780893)
    (2.228635,2.773605)
    (2.255838,2.766883)
    (2.283040,2.760699)
    (2.310243,2.755027)
    (2.337446,2.749842)
    (2.364648,2.745121)
    (2.391851,2.740845)
    (2.419053,2.736993)
    (2.446256,2.733548)
    (2.473459,2.730492)
    (2.500661,2.727812)
    (2.527864,2.725491)
    (2.555066,2.723516)
    (2.582269,2.721875)
    (2.609471,2.720555)
    (2.636674,2.719547)
    (2.663877,2.718838)
    (2.691079,2.718419)
    (2.718281,2.718282)
};

\addplot[
    thick,
    densely dotted,
    brown,
    smooth,
]
coordinates {
    (2.718282,2.718282)
    (2.609601,2.800000)
    (2.380722,3.000000)
    (2.193328,3.200000)
    (2.037254,3.400000)
    (1.905354,3.600000)
    (1.792474,3.800000)
    (1.694808,4.000000)
    (1.609490,4.200000)
    (1.534325,4.400000)
    (1.469338,4.600000)
    (1.420204,4.800000)
    (1.381966,5.000000)
    (1.313859,5.500000)
    (1.267949,6.000000)
    (1.208712,7.000000)
    (1.171573,8.000000)
    (1.127017,10.000000)
    (1.101021,12.000000)
    (1.077384,15.000000)
    (1.055728,20.000000)
    (1.043561,25.000000)
    (1.035760,30.000000)
    (1.026334,40.000000)
};
\addlegendentry{Alg~\cite{angelopoulos2025learning}}

\addplot[
    only marks,
    mark=*,
    mark size=1.8pt,
    violet,
]
coordinates {
    ({exp(1)}, {exp(1)})
};
\addlegendentry{Comp alg}
\end{axis}
\end{tikzpicture}
\caption{Online bidding}
\label{subfig:consistency-robustness-tradeoff-bidding}
\end{subfigure}
\hfill
\begin{subfigure}{0.49\textwidth}
\centering
\begin{tikzpicture}
\pgfmathsetmacro{\Btwo}{4.59112147667} 

\begin{axis}[
    width=\textwidth,
    xlabel={Consistency $\chi$},
    ylabel={Robustness $\rho$},
    xmin=1, xmax=4.8,
    xmode=log,
    log basis x={2.718281828},
    xtick={1, 2, 3, 4, 4.59112147667},
    xticklabels={$1$, $2$, $3$, $4$, $\rho_{\mathrm{LS}}^*$},
    ymin=4.3, ymax=50,
    ymode=log,
    log basis y={2.718281828},
    ytick={4.59112147667, 10, 20, 40},
    yticklabels={$\rho_{\mathrm{LS}}^*$, $10$, $20$, $40$},
    grid=both,
    legend pos=north east,
    samples=300,
    font=\footnotesize,
    set layers
]


\addplot[
    on layer=axis foreground,
    thick,
    solid,
    blue,
    variable=\z,
    domain=0.0415:0.58774315916,
]
(
    {1+(exp(2*\z)-1-(2*\z))/(\z*(1+exp(\z)))},
    {1+(1+exp(\z))/(\z)}
);


\addplot[
    on layer=axis foreground,
    thick,
    solid,
    blue,
    smooth,
    forget plot,
]
coordinates {
    (1.646700,5.763852)
    (1.660611,5.707114)
    (1.674589,5.652909)
    (1.688649,5.601101)
    (1.702804,5.551566)
    (1.717069,5.504185)
    (1.731457,5.458849)
    (1.745985,5.415458)
    (1.760664,5.373917)
    (1.775511,5.334137)
    (1.790540,5.296037)
    (1.805764,5.259540)
    (1.821201,5.224573)
    (1.836864,5.191069)
    (1.852771,5.158965)
    (1.868938,5.128202)
    (1.885381,5.098724)
    (1.902120,5.070478)
    (1.919172,5.043415)
    (1.936557,5.017489)
    (1.954297,4.992655)
    (1.972412,4.968872)
    (1.990926,4.946102)
    (2.009862,4.924308)
    (2.029248,4.903453)
    (2.049111,4.883507)
    (2.069480,4.864436)
    (2.090387,4.846213)
    (2.111867,4.828808)
    (2.133957,4.812195)
    (2.156696,4.796349)
    (2.180129,4.781246)
    (2.204302,4.766863)
    (2.229268,4.753179)
    (2.255084,4.740173)
    (2.281813,4.727825)
    (2.309524,4.716117)
    (2.338295,4.705031)
    (2.368212,4.694550)
    (2.399373,4.684659)
    (2.431886,4.675341)
    (2.465875,4.666582)
    (2.501483,4.658369)
    (2.538871,4.650687)
    (2.578229,4.643525)
    (2.619777,4.636870)
    (2.663774,4.630710)
    (2.710531,4.625035)
    (2.760421,4.619834)
    (2.813902,4.615097)
    (2.871545,4.610813)
    (2.934072,4.606975)
    (3.002425,4.603573)
    (3.077863,4.600599)
    (3.162133,4.598045)
    (3.257779,4.595902)
    (3.368750,4.594164)
    (3.501770,4.592823)
    (3.670121,4.591874)
    (3.908562,4.591309)
    (4.591121,4.591121)
};
\addlegendentry{Our alg}


\addplot[
    thick,
    dashed,
    teal,
    variable=\t,
    domain=0.005:1.1,
]
(
    {1 + 2*\t*(\t+2)^2/(4+3*\t-\t^3)},
    {max(\Btwo, 1 + 4*(\t^2+\t+1)/(\t*(4+3*\t-\t^3)))}
);
\addlegendentry{Our LB}


\addplot[
    thick,
    densely dotted,
    brown,
    smooth,
]
coordinates {
    (4.591121,4.591121)
    (4.359176,4.700000)
    (4.168129,4.800000)
    (3.838294,5.000000)
    (3.232367,5.500000)
    (2.824295,6.000000)
    (2.533882,6.500000)
    (2.318252,7.000000)
    (2.022091,8.000000)
    (1.830167,9.000000)
    (1.696714,10.000000)
    (1.524655,12.000000)
    (1.380779,15.000000)
    (1.260025,20.000000)
    (1.196958,25.000000)
    (1.158361,30.000000)
    (1.132351,35.000000)
    (1.113651,40.000000)
    (1.088580,50.000000)
    (1.053263,80.000000)
    (1.042070,100.000000)
    (1.027576,150.000000)
    (1.020509,200.000000)
};
\addlegendentry{Alg~\cite{angelopoulos2025learning}}


\addplot[
    only marks,
    mark=*,
    mark size=1.8pt,
    violet,
]
coordinates {
    (\Btwo,\Btwo)
};
\addlegendentry{Comp alg}

\end{axis}
\end{tikzpicture}
\caption{Linear search}
\label{subfig:linear-search-consistency-robustness-search}
\end{subfigure}
\caption{Robustness-consistency trade-offs of \textbf{(a)} online bidding and \textbf{(b)} linear search. All axes use logarithmic scaling.
\textbf{(a)}~The solid blue curve depicts our Pareto-optimal trade-off (\Cref{thm:main-bidding}), while the dash-dotted red and dashed orange curves represent the upper and lower bounds from \cite{shin25}, respectively. The dotted brown curve shows the trade-off achieved by \cite{angelopoulos2025learning}. The violet point corresponds to the classical $e$-competitive algorithm without learning augmentation.
\textbf{(b)}~The solid blue curve depicts the trade-off achieved by our algorithm \Cref{thm:lin:Kbranch_construction}), while the dashed teal curve represents our lower bound(\Cref{thm:lin:lb}). The dotted brown curve shows the upper bound from \cite{angelopoulos2025learning}. The violet point corresponds to the classical linear search algorithm with competitive ratio $\rho_{\mathrm{LS}}^* := 1 + 1/W_0(e^{-1})\approx4.59$.
}
\label{fig:consistency-robustness-tradeoff}
\end{figure}
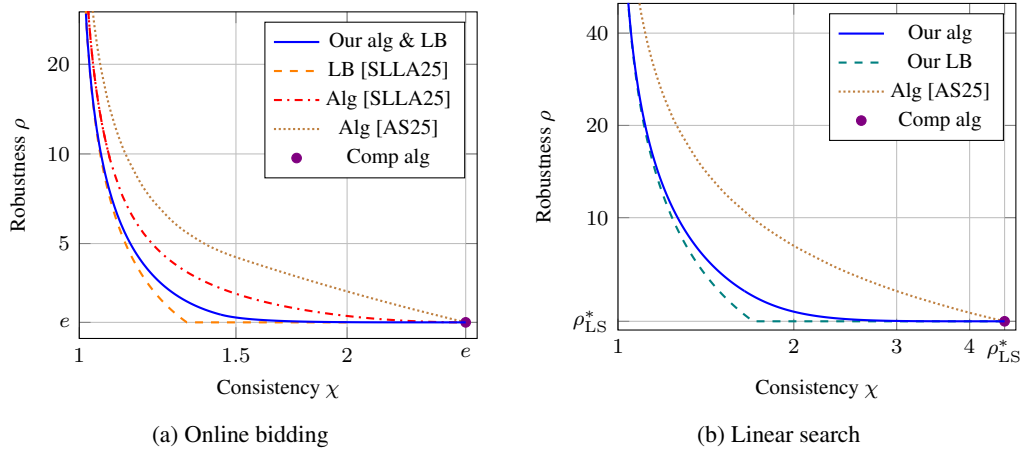

%% file: arxiv_include/2_preliminaries.tex
\section{Preliminaries}\label{sec:prelim}
This section formally defines the problems considered in this paper, along with the performance measures and the notation used throughout this paper.

\paragraph{Notation} For a real function $f$ and a real number $x$, we denote $f(x+) = \lim_{y \downarrow x} f(y)$, and $f(x-) = \lim_{y \uparrow x} f(y)$. For an event $\mathcal{E}$, we denote the indicator function with $\mathbf{1}_{\mathcal{E}}$. For a function $G$ and a subset of its domain $I$, we denote the restriction of $G$ over $I$ as $\left. G \right\rvert_{I}$.
We denote by $W_0$ the principal branch of the lambert $W$ function.

\paragraph{Problem definitions}
In the \emph{online bidding} problem, there is an unknown target value $T>0$.
A bidder repeatedly submits (strictly increasing) bids $b\in \bbR_{>0}$ until a bid $b \ge T$ is submitted.
Prior to termination, the bidder receives no information about $T$ other than the fact that all previous bids were smaller than $T$.
The cost is defined as the sum of all submitted bids up to and including the first bid of value at least
$T$.
Naturally, the objective is to minimize the cost.
In the \emph{linear search} problem,
there is an unknown target point $T\in\bbR\setminus\{0\}$, 
and the searcher starts at the origin.
One can view the search as a sequence of signed \emph{excursion endpoints}. In a single excursion to a signed endpoint $z \in \bbR \setminus \{0\}$, the searcher starts at the origin and continuously moves toward $z$. If the searcher reaches the target on the way, then the search immediately stops there; on the other hand, the searcher returns to the origin, ending this excursion. The cost is defined as the total distance traveled until termination. It is also without loss of generality to assume
that consecutive excursions have opposite directions (i.e., signs) and that the magnitudes of excursion endpoints are strictly increasing in each direction.

In the \emph{learning-augmented} setting, the instance additionally contains a prediction for the unknown target.
For online bidding, the prediction is a value $P>0$ on the target value $T$; for linear search, the prediction is a point $P \in \bbR \setminus \{0\}$ on the target point $T$.
Note that the prediction can be inaccurate.

From this point, we focus on online bidding due to space constraints.  
The analogous discussion for linear search is deferred to Appendix~\ref{app:lin:prelim}.

\paragraph{Algorithms and strategies}
Note that an algorithm for online bidding can be represented by an infinite sequence of increasing real numbers.
However, since the algorithm should be competitive against every possible $T > 0$ which can be any positive real number, it is difficult to define the \emph{first} bid submitted by the algorithm.
To address this issue, one standard assumption is bounding $T$ from below by some constant (say, $1$)~\cite{angelopoulos2024online,angelopoulos2025learning,shin25}.

In this paper, however, we introduce a surrogate model for an algorithm, called a \emph{strategy}, defined as a bi-infinite sequence $ b = \{b_n\}_{n \in \bbZ}$ of strictly increasing bids $0 < \cdots < b_{-1} < b_0 < b_1 < \cdots$.
The cost incurred by a deterministic strategy $b$ given a target value $T > 0$, denoted by $\cost_b(T)$, is then defined as the prefix sum of bids up to and including the first bid at least $T$, i.e., 
\[
    \textstyle 
    \cost_b(T) := \sum_{n \leq n_\ast} b_n, \text{ where $n_\ast := \min \{ n \in \bbZ : b_n \geq T \}$}.
\]
A randomized strategy $B$ is naturally defined as a probability distribution over deterministic strategies.
The expected cost incurred by $B$ given $T$ is denoted by $\cost_B(T)$.

In what follows, we study the performance guarantees of strategies with the cost function defined as above.
The main reason is that it provides substantial simplicity in our analysis.
In fact, even though a strategy cannot serve as an algorithm by itself due to its bi-infinite nature, one can easily generate from a strategy an algorithm which incurs no more than the strategy (in expectation) in the setting where the target value is bounded from below by some constant.
Moreover, we can also argue that any lower bound on the trade-off for strategies also serves as a valid lower bound for algorithms since a strategy can be viewed as an algorithm given a sufficiently large prediction. A detailed discussion on this justification is provided in Appendix~\ref{app:str2alg}.

\paragraph{Normalization and performance measures}
As the universe of $T$ is $(0, \infty)$, the problem is homogeneous under scaling.
We therefore assume from now that the prediction is given as $1$ (by dividing the whole universe by $P$).
Moreover, as a strategy is bi-infinite, we can also freely shift the indices of the strategy.
Hence, we assign index 0 for the first bid of value at least $1 = P$.

To measure the performance of learning-augmented algorithms, we adopt the \emph{robustness-consistency analysis} \cite{lykouris2021competitive,mitzenmacher2022algorithms}.
We say a (randomized) strategy $B$ is $\chi$-\emph{consistent} if $\cost_B(1) \leq \chi$ and $\rho$-\emph{robust} if $\cost_B(T) \leq \rho \, T$ for every $T > 0$.
A pair $(\rho, \chi)$ is \emph{Pareto-optimal} if there exists a $\rho$-robust $\chi$-consistent strategy, but no $\rho'$-robust $\chi$-consistent strategy for any $\rho' < \rho$.

%% file: arxiv_include/3_optimal_bidding.tex
\section{Pareto-optimal trade-off for online bidding} \label{sec:bidding_optimal}
In this section, we present a Pareto-optimal learning-augmented strategy
for online bidding.
To this end, we introduce in \Cref{subsec:profile_function_framework} a novel framework called a \emph{bidding profile} for representing the distribution over bids made by a strategy.
In \Cref{subsec:bid:tight}, we construct a bidding profile such that the strategy driven by the profile attains the best-possible trade-off among all profile-driven strategies.
Finally in \Cref{subsec:bid:alg2profile}, we prove that any 
bidding strategy
can be reduced to a bidding profile, implying the Pareto-optimality of the strategy driven by our bidding profile.

\subsection{Bidding profiles} \label{subsec:profile_function_framework}

To motivate our framework, we revisit the well-known $e$-competitive algorithm for online bidding~\cite{chrobak2008incremental}. Under the assumption that $T \geq 1$, it samples $U \sim \Unif(0, 1]$ at the beginning, submits the first bid $e^U$, and repeatedly increases the bid by a factor $e$ until termination.
Note that we can represent the random execution of the strategy corresponding to this algorithm by a single non-decreasing function $G(x) := e^x$ as follows:
For a shared random source $U \sim \Unif(0, 1]$, the strategy is defined as $\{G(n + U)\}_{n \in \bbZ}$.

Let us now analyze the competitive ratio of this strategy, with the assumption lifted.
For the target $T > 0$, we define $x := G^{-1}(T) = \ln T$ and write $x = q + r$ for some (uniquely determined) $q \in \bbZ$ and $r \in (0, 1]$.
Note that the last bid submitted by the strategy is $G(q+U)$ if $U \geq r$, and $G(q + 1 + U)$ if $U < r$.
Hence, the expected sum of total bids submitted by the strategy is computed as follows:
\begin{align*}
    \int_0^r \sum_{n \leq q+1} G(n + u) \, \dd u
    + \int_{r}^1 \sum_{n \leq q} G(n + u) \, \dd u
    &
    = \int_0^1 \sum_{n \leq q} G(n + u) \, \dd u 
    + \int_0^r G(q + 1 + u) \, \dd u
    \\&
    = \int_{-\infty}^{q+1+r} G(t) \, \dd t
    = \int_{-\infty}^{x + 1} G(t) \, \dd t.
\end{align*}
Observe that, for any $T > 0$ with $x = G^{-1}(T)$, we have
\begin{equation} \label{eq:bid:compexample}
    \int_{-\infty}^{x+1} G(t) \, \dd t = e \, G(x) = e \, T
\end{equation}
by the choice of function $G$, showing the $e$-competitiveness of the strategy.

In what follows, we generalize $G$ in the above strategy to an arbitrary non-decreasing positive function, and focus on the family of bidding strategies driven by such functions $G$.
More precisely, given a non-decreasing function $G$, 
the strategy is defined as $\{ G(n + U) \}_{n \in \bbZ}$ for a shared random source $U \sim \Unif(0, 1]$.
We highlight that this strategic framework already captures many existing algorithms with or without learning augmentation~\cite{chrobak2008incremental,anand2021regression,angelopoulos2024online,shin25,angelopoulos2025learning}.
For a simple example, notice that the strategy corresponding to the classical deterministic $4$-competitive algorithm can be represented by $G(x) := 2^{\lceil x \rceil - 1}$.

In the learning-augmented setting, we are given a prediction $P$ on $T$. Recall we assume $P = 1$.
Moreover, we horizontally shift $G$ so that $\left. G\right\rvert_{(0, 1]}$ represents the strategy's bid that becomes at least the prediction for the first time.
The robustness and consistency of the strategy driven by $G$ can be expressed analogously to \eqref{eq:bid:compexample}, leading us to the following definition of a \emph{bidding profile}:

\begin{definition}[Bidding profile]
    Given $\rho$ and $\chi$ such that $1 < \chi \le \rho$, a non-decreasing, left-continuous function
    $G : \bbR \to (0, \infty)$ is called a $(\rho, \chi)$-\emph{bidding profile} if
    \begin{itemize}
        \item (Offset condition) $G(x) < 1$ for every $x < 0$, and $G(x) \ge 1$ for every $x > 0$;

        \item (Robustness condition)
        \(
            \int_{-\infty}^{x+1}G(t) \dd t \le \rho \, G(x) \; \text{for every $x \in \bbR$};
        \)
        \item (Consistency condition)
        \(
            \int_{-\infty}^{1}G(t) \dd t \le \chi.
        \)
    \end{itemize}
\end{definition}

\begin{restatable}{remark}{rmkprofile}
    Note that $G(0) \le 1$, from its left-continuity and the offset condition. Also, the strict positivity of $G$ can be deduced from the robustness condition. \label{rmk:bid:positivity}
\end{restatable}

Given a bidding profile $G$, the strategy driven by $G$ is expressed by $\{ G(n + U) \}_{n \in \bbZ}$,
where $U \sim \Unif(0, 1]$ is a random source shared across $n \in \bbZ$.

As $G$ is non-decreasing, this bidding sequence is indeed well defined.
Denoting by $\cost_G(T)$ the expected cost incurred by the strategy driven by $G$ given $T$,
the next lemma provides a formula for $\cost_G(T)$,
implying the robustness and consistency of this strategy.
\begin{restatable}{lemma}{lembidcost}
\label{lem:bid_cost}
Let $G$ denote a $(\rho,\chi)$-bidding profile.
Then, for any $T > 0$,
\[
    \textstyle
    \tau(T):=\sup\{t\in\mathbb{R}:G(t)<T\}<\infty,
    \;\;\text{and}\;\;
    \cost_G(T)=\int_{-\infty}^{\tau(T)+1}G(t) \dd t
\]
\end{restatable}
Following is a direct corollary from the above lemma and the definition of a bidding profile.
\begin{corollary}\label{cor:G_bidding}
    The strategy driven by a $(\rho, \chi)$-bidding profile is $\rho$-robust and $\chi$-consistent.
\end{corollary}
\subsection{Pareto-optimal bidding profile}
\label{subsec:bid:tight}

We now construct a bidding profile which gives the Pareto-optimal trade-off among profile-driven strategies. We present the following theorem, implying the first half of \Cref{thm:main-bidding} together with \Cref{cor:G_bidding}.
\begin{restatable}{theorem}{thmbiddingub}
    For any $s \in (0, 1]$, there exists a $\left(\rho(s), \chi(s)\right)$-bidding profile, where
    \[
    \rho(s):=\frac{e^s}{s}
    \quad \text{ and } \quad
    \chi(s):=
    \begin{cases}
    (e^s-1)/s & (0<s\le\ln2),\\[1mm]
    \xi(s)/s & (\ln2\le s\le1),
    \end{cases}
    \]
    where $\xi(s)$ is the unique solution $\xi \in [1, e]$ satisfying $\xi(2-\ln\xi)=e^s$ for any $s \in [\ln 2, 1]$.
    \label{thm:bidding_ub}
\end{restatable}

In the interest of space, we sketch the proof here, while deferring the full proof to Appendix~\ref{app:bid:tight}.
\begin{proof}[Proof sketch]
For $s = 1$, it is trivial to see that the bidding profile $G(x) = e^x$ induces the classical $e$-competitive algorithm.
We thus fix any $s \in (0, 1)$,
and write $\rho := \rho(s)$ and $\chi := \chi(s)$ for simplicity.
We first define a bidding profile $G$ for $x > 0$ as follows:
\begin{equation} \label{eq:optbid:Gpositive}
    G(x):=\begin{cases}
        \max(1,s\chi \cdot e^{s(x-1)}), & \text{ for } 0<x\le 1,\\
        \max(1,s\chi)\cdot e^{s(x-1)},& \text{ for } x>1.
    \end{cases}
\end{equation}
It remains to define $G$ on $x\le0$, which is defined inductively as follows. 
Let $G(0):=\frac{\chi}{\rho} \leq 1$.
For any non-positive integer $k \in \bbZ_{\le 0}$, given $G$ defined on $[k, k+1)$, we define $G$ on $[k - 1, k)$ as
\begin{equation} \label{eq:optbid:Ginduct}
    G(x):=G(k)-\frac{1}{\rho}\int_{x+1}^{k+1}G(t)\dd t 
    \;\;\text{for any $x \in [k-1, k)$.}
\end{equation}
Intuitively, the above expression computes, given the distribution of bid $G(k+U)$, the optimal distribution of bid $G(k-1+U)$ satisfying the desired robustness and consistency.
See \Cref{fig:bidding-profile-regimes} for illustrations of the constructed profile $G$ for several representative values of the parameter $s$.

The key technical part is to show that $G$ is non-negative.
For the moment, let us assume that this statement is true. We can then show that $G$ is indeed a $(\rho, \chi)$-bidding profile.
Observe that $G$ is non-decreasing due to \eqref{eq:optbid:Gpositive} and \eqref{eq:optbid:Ginduct} along with the non-negativity assumption and $G(0) \leq 1$, where the offset condition also follows.
The left-continuity of $G$ is guaranteed by the fact that $G$ is continuous on $x > 0$ and $x \leq 0$ again due to \eqref{eq:optbid:Gpositive} and \eqref{eq:optbid:Ginduct}, respectively.
For the remaining two conditions, let us first argue the robustness condition for $x < 0$.
For any non-positive integer $k \in \bbZ_{\le 0}$ and $x \in [k-1, k)$, we derive from \eqref{eq:optbid:Ginduct}
\[
    \textstyle 
    G(x) = G(x) - G(k) + \sum_{n \in \bbZ_{\le k}}(G(n)-G(n-1)) = \frac{1}{\rho}\int_{-\infty}^{x+1}G(t)\dd t,
\]
which shows the robustness condition for $x < 0$. 
Due to the left-continuity, the equation also holds for $x = 0$, implying the consistency condition together with the choice of $G(0) = \frac{\chi}{\rho}$.
The robustness condition for $x > 0$ can be shown analogously to the standard analysis of the classical $e$-competitive algorithm.

To show the non-negativity of $G$, we define the following functional operator.
\begin{restatable}{definition}{defncF}\label{defn:cF}
    Given $\rho > 0$, let $\cF$ be a functional operator on an integrable, non-decreasing, non-negative function $H : (-\infty, 1] \to [0, \infty)$, giving a function $\cF H : (-\infty, 1] \to [0, \infty)$ such that
    \begin{align*}
    \cF H (x) := \begin{cases}
        H(x), & 0<x\le 1,\\
        \frac{1}{\rho}\int_{-\infty}^{x+1} H(t)\dd t, &x \le 0. 
    \end{cases}
    \end{align*}
\end{restatable}
We claim that our bidding function $G$ satisfies $G|_{(-\infty, 0]} \equiv G^\ast|_{(-\infty, 0]}$, where $G^\ast : (-\infty, 1] \to [0, \infty) $ is a function satisfying (i) $G^\ast = \cF G^\ast$, (ii) $\int_{-\infty}^1 G^\ast(t) \dd t = \chi$, and (iii) $G^\ast |_{(0, 1]} \equiv \phi|_{(0, 1]}$ where $\phi(x) := \max(1, s\chi \cdot e^{s(x-1)})$.
First, we can derive $G(0) = G^\ast(0) = \frac{\chi}{\rho}$ from (i) and (ii).
Next, for any non-positive integer $k \in \bbZ_{\le 0}$ and $x \in [k-1, k)$, 
\begin{align*}
    \textstyle
    G^\ast(x) 
    &
    \textstyle
    = \frac{1}{\rho} \int_{-\infty}^{x+1} G^\ast(t) \dd t
    = \frac{1}{\rho} \left( \int_{-\infty}^{k+1} G^\ast(t) \dd t - \int_{x+1}^{k+1} G^\ast(t) \dd t \right)
    \\&
    \textstyle
    = G^\ast(k) - \frac{1}{\rho} \int_{x+1}^{k+1} G^\ast(t) \dd t = G(x),
\end{align*}
where the first and third equalities are from (i), and the last equality is by the inductive hypothesis.
This completes the proof of the claim if $G^\ast$ exists.

It thus remains to show the existence of $G^\ast$.
We exploit the \emph{order-preserving} property of $\cF$: if $A \le B$ for two functions $A,B$ (i.e., $A(x)\le B(x)$ for all $x \in (-\infty, 1]$), then $\cF A \le \cF B$.
Let $G^{(0)}$ be a function on $(-\infty, 1]$ such that $G^{(0)}(x):=0$ for $x\le 0$ and $G^{(0)}(x):=\phi(x)$ for $0<x\le 1$, and for $n \in \bbZ_{\ge 0}$, let $G^{(n+1)}(x) := \cF G^{(n)}(x)$.

Since $G^{(0)} \le G^{(1)}$, the order-preserving property implies $G^{(n)}\le G^{(n+1)}$ for all $n \in \bbZ_{\ge 0}$.
We then show that, for all $n \in \bbZ_{\ge 0}$, $G^{(n)}$ is bounded from above by $H$ defined as $H(x) := Ke^{cx}$ for $x \leq 0$ and $H(x) := \phi(x)$ for $0 < x \leq 1$, 
for some $c \in (s,1)$ and sufficiently large $K>0$.
Indeed, $\cF H \le H$, which implies
\(
G^{(n)} \le \cF^{(n)} G^{(0)} \le \cF^{(n)} H \le H,
\)
where the second inequality follows from the definition of $G^{(0)}$ and $H$.
By the monotone convergence theorem, $\{G^{(n)}\}_{n \in \bbZ_{\ge 0}}$ converges.
Simple calculation shows that the converged function satisfies the remaining conditions for $G^\ast$, completing the proof.
\end{proof}

We now argue that the trade-off of \Cref{thm:bidding_ub} is indeed the best possible among profile-driven strategies.
Let us say a bidding profile $G$ is \emph{induced by $\phi: (0, 1] \to [1, \infty)$} if $\left. G \right\rvert_{(0, 1]} \equiv \phi$.
We further define a \emph{tight} bidding profile.
\begin{definition}[Tight bidding profile]
    We say a $(\rho,\chi)$-bidding profile $G$ is \emph{tight} if it further satisfies
    \begin{align*}
    \textstyle
    \int_{-\infty}^{x+1}G(t) \dd t = \rho \, G(x) \;\;\text{for all $x \le 0$},
    \quad \text{and} \quad 
    \int_{-\infty}^{1} G(t) \dd t = \chi.
    \end{align*}
\end{definition}

We then show that it suffices to consider tight bidding profiles.
\begin{restatable}{lemma}{lembidtightening}
    Suppose there exists a $(\rho, \chi)$-bidding profile $G$ for fixed positive numbers $\rho > \chi > 1$. Let $\phi: (0,1]\to [1,\infty)$ denote the function that induces $G$.
    Then, there exists a tight $(\rho, \chi')$-bidding profile $G'$ induced by $\phi$ such that $\chi' \le \chi$.
    \label{lem:bid:tightening}
\end{restatable}

Finally, the following theorem implies the optimality of our bidding profile.
\begin{restatable}{theorem}{thmbiddinglb}
    For any $0<s \le1$, if there exists a tight $\left(\rho(s), \chi\right)$-bidding profile, then $\chi \ge \chi(s)$.
    \label{thm:bidding_lb}
\end{restatable}
Again, all proofs are deferred to Appendix~\ref{app:bid:tight} due to the space constraint.

\input{arxiv_include/3a_tikz_vis}

\subsection{Generality of bidding profiles}
\label{subsec:bid:alg2profile}

To complete the proof of \Cref{thm:main-bidding}, we show the following lemma.

\begin{restatable}{lemma}{lembiddinggconvert}
    For any fixed positive real numbers $\rho$ and $\chi$ such that $\rho > \chi > 1$, there exists a $\rho$-robust $\chi$-consistent online bidding strategy if and only if there exists a $(\rho,\chi)$-bidding profile $G$.
    \label{lem:bidding_g_convert}
\end{restatable}

Note that the second half of \Cref{thm:main-bidding} then immediately follows from \Cref{thm:bidding_lb}.
Due to the space constraint, we only provide a sketch of its proof, while deferring the full proof to Appendix~\ref{app:str2prof}.

\begin{proof}[Proof sketch]
Recall that a bidding strategy is a probability distribution over bi-infinite increasing sequences in $(0, \infty)$.
Hence, given a $\rho$-robust $\chi$-consistent bidding strategy $\cA$ with prediction $P = 1$, by simply accumulating the probability measure value-wise, we can obtain an \emph{aggregate measure}~$\mu$; informally speaking, for each value $t \in (0, \infty)$, $\mu(t)$ corresponds to the marginal probability that $\cA$ ever submits $t$ as its bid value when $\cA$ never terminates.

We then define the \emph{inverse profile} $G : \bbR \to (0, \infty)$ from $\mu$ as follows:
\[
    G(x) := \begin{cases}
        \sup\{0< t<1:\ \mu([t,1)) \ge -x\}, & \text{ if $x < 0$;} \\
        \inf\{t\ge1:\ \mu([1,t])\ge x\}, & \text{ if $x \geq 0$.}
    \end{cases}
\]
Intuitively, $G$ can be seen as a bidding profile whose driven strategy $\cB$ induces the same aggregate measure $\mu$ as $\cA$.

To show that $G$ is actually a $(\rho, \chi)$-bidding profile, we argue that the expected cost incurred by $\cB$ is always no worse than that by $\cA$ for every target value $T > 0$.
On the one hand, the expected costs incurred before submitting at least $T$ are the same between $\cA$ and $\cB$ since they have the same aggregate measure.
On the other hand, the distribution of the last bid (which is at least $T$) can be viewed as a packing of size $1$ in $\mu([T, \infty))$, and by the definition of $G$, $\cB$ packs this distribution as small as possible with respect to $\mu([T, \infty))$.
These together imply that $\cB$ incurs no worse than $\cA$ in expectation for every $T > 0$.
\end{proof}

%% file: arxiv_include/3a_tikz_vis.tex
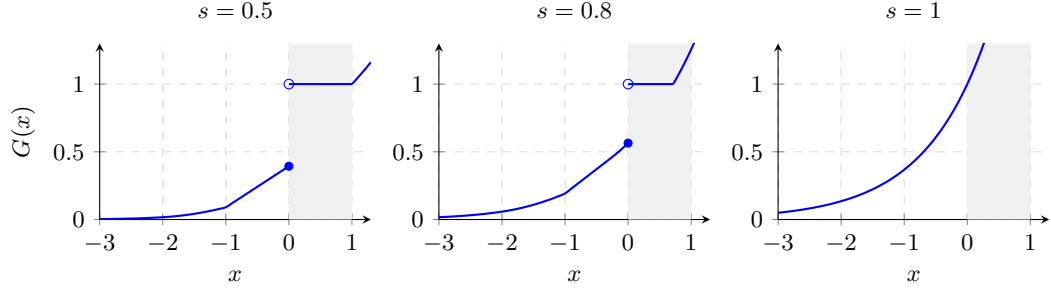
\begin{figure}[t]
\centering
\begin{tikzpicture}
\begin{groupplot}[
    group style={group size=3 by 1, horizontal sep=0.9cm},
    width=0.37\textwidth,
    height=0.28\textwidth,
    xmin=-3, xmax=1.3,
    ymin=0, ymax=1.3,
    axis lines=left,
    xlabel={$x$},
    xtick={-3,-2,-1,0,1},
    ytick={0,0.5,1.0},
    tick label style={font=\small},
    label style={font=\small},
    title style={font=\small},
    xmajorgrids=true,
    ymajorgrids=true,
    major grid style={dashed, gray!30},
    axis on top=false,
    clip=true,
]

\nextgroupplot[
    title={$s=0.5$},
    ylabel={$G(x)$},
]

\addplot[draw=none, fill=gray!12] coordinates {(0,0) (1,0) (1,1.3) (0,1.3)} \closedcycle;

\addplot[thick, blue, smooth] coordinates {
(-3.0,0.003017) (-2.9,0.003577) (-2.8,0.004248) (-2.7,0.005058) (-2.6,0.006035)
(-2.5,0.007206) (-2.4,0.008599) (-2.3,0.010243) (-2.2,0.012165) (-2.1,0.014393)
(-2.0,0.016955) (-1.9,0.020151) (-1.8,0.024266) (-1.7,0.029301) (-1.6,0.035255)
(-1.5,0.042129) (-1.4,0.049923) (-1.3,0.058637) (-1.2,0.068270) (-1.1,0.078823)
(-1.0,0.090296) (-0.9,0.120531) (-0.8,0.150857) (-0.7,0.181184) (-0.6,0.211510)
(-0.5,0.241837) (-0.4,0.272163) (-0.3,0.302490) (-0.2,0.332816) (-0.1,0.363143)
(0.0,0.393469)
};

\addplot[thick, blue, domain=0.001:1, samples=2] {1};
\addplot[thick, blue, domain=1:1.3, samples=100] {exp(0.5*(x-1))};

\addplot[blue, only marks, mark=*, mark size=1.5pt] coordinates {(0,0.393469)};
\addplot[blue, only marks, mark=o, fill=white, mark size=1.8pt] coordinates {(0,1)};

\nextgroupplot[
    title={$s=0.8$},
]

\addplot[draw=none, fill=gray!12] coordinates {(0,0) (1,0) (1,1.3) (0,1.3)} \closedcycle;

\addplot[thick, blue, smooth] coordinates {
(-3.0,0.017085) (-2.9,0.019308) (-2.8,0.021825) (-2.7,0.024684) (-2.6,0.027930)
(-2.5,0.031611) (-2.4,0.035772) (-2.3,0.040460) (-2.2,0.045722) (-2.1,0.051605)
(-2.0,0.058162) (-1.9,0.065717) (-1.8,0.074564) (-1.7,0.084703) (-1.6,0.096134)
(-1.5,0.108858) (-1.4,0.122874) (-1.3,0.138181) (-1.2,0.154792) (-1.1,0.172786)
(-1.0,0.192279) (-0.9,0.228147) (-0.8,0.264093) (-0.7,0.300039) (-0.6,0.335986)
(-0.5,0.371932) (-0.4,0.407878) (-0.3,0.443825) (-0.2,0.480826) (-0.1,0.520873)
(0.0,0.564255)
};

\addplot[thick, blue, domain=0.001:0.715312, samples=2] {1};
\addplot[thick, blue, domain=0.715312:1.3, samples=150] {1.2557715 * exp(0.8*(x-1))};

\addplot[blue, only marks, mark=*, mark size=1.5pt] coordinates {(0,0.564255)};
\addplot[blue, only marks, mark=o, fill=white, mark size=1.8pt] coordinates {(0,1)};

\nextgroupplot[
    title={$s=1$},
]

\addplot[draw=none, fill=gray!12] coordinates {(0,0) (1,0) (1,1.3) (0,1.3)} \closedcycle;
\addplot[thick, blue, domain=-3:1.3, samples=200] {exp(x)};

\end{groupplot}
\end{tikzpicture}
\caption{Representative bidding profiles $G$ for three regimes of $s$ which determines the trade-off. The shaded strip highlights the interval $(0,1]$.}
\label{fig:bidding-profile-regimes}
\end{figure}

%% file: arxiv_include/4_linear_search.tex
\section{Extension to linear search} \label{sec:linsearch}

In this section, we present a learning-augmented strategy for the linear search problem. We introduce in \Cref{subsec:excursion_profile} the framework called an \emph{excursion profile}, 
adapted from a bidding profile
to describe a sequence of excursions. 
In \Cref{subsec:lin:tight}, we present 
a near-optimal excursion profile whose driven strategy shows the robustness-consistency trade-off presented in \Cref{subfig:linear-search-consistency-robustness-search}.

\subsection{Excursion profiles} \label{subsec:excursion_profile}
Since linear search is also scale-invariant, we assume that the prediction is located at coordinate $+1$. Instead of a single profile, we assign a pair $G_{\pm} = (G_+, G_-)$ of non-decreasing, non-negative functions as an \emph{excursion profile}. 
The strategy driven by $G_\pm$ naturally corresponds to the following
excursion sequence with a common random source $U \sim \Unif(0, 1]$:
\[
\cdots +G_{+}(n + U) \to  -G_{-}(n + U) \to +G_{+}(n+1+U) \to  -G_{-}(n+1+U), \cdots
\]
Moreover, we define $\left. G_{+}\right\rvert_{(0, 1]}$ as the first excursion 
passing through the prediction point. 
This leads us to the following definition.

\begin{definition}[Excursion profile]
    For $\rho$ and $\chi$ with $\rho \ge \chi > 0$, a pair $G_\pm$ of non-decreasing, left-continuous functions $G_{+}, G_- : \bbR \to (0, \infty)$ is called a $(\rho, \chi)$-\emph{excursion profile} if
    \begin{itemize}
        \item (Plus-offset condition) $G_{+}(x) \ge 1$ for any $x > 0$, and $G_{+}(x) < 1$ for any $x < 0$;
        \item (Robustness condition)
        $C_{+}(x) \le \rho \, G_{+}(x)$ and $C_{-}(x) \le \rho \, G_{-}(x)$ for all $x \in \bbR$;
        \item (Consistency condition) $ C_{+}(0) \leq \chi$,
    \end{itemize}
    where \(
            \textstyle
            C_{+}(x) := \int_{-\infty}^{x} G_{+}(t) \dd t + \int_{-\infty}^{x} G_{-}(t)\dd t
        \) and
        \(
            \textstyle
            C_{-}(x) := \int_{-\infty}^{x+1} G_{+}(t) \dd t + \int_{-\infty}^{x} G_{-}(t)\dd t.
        \)
\end{definition}

We can easily adapt the proof of \Cref{lem:bid_cost} to derive the following lemma whose proof is deferred to \Cref{app:subsec:excursion_profile}.
\begin{restatable}{lemma}{corlsbid}
    The strategy driven by a $(\rho, \chi)$-excursion profile $G_{\pm}$ is $(1+2\rho)$-robust and $(1+2\chi)$-consistent.
\end{restatable}

\subsection{Near-optimal excursion profile} \label{subsec:lin:tight}

In this section, we construct a near-optimal excursion profile, using the methods derived from \Cref{subsec:bid:tight}. 
The proofs can be found in Appendix~\ref{app:lin:ub}.

Let us say an excursion profile $G_{\pm}$ is \emph{induced by} $\psi : (0, 1] \to [1, \infty)$ if $\left. G_{+} \right\rvert_{(0, 1]} \equiv \psi$. We will consider \emph{tight} excursion profiles induced by $\psi$ without degrading performance.

\begin{definition}[Tight excursion profile]
    We say a $(\rho, \chi)$-excursion profile $G_{\pm}$ is \emph{tight} if it further satisfies 
    $C_+ (x) = \rho \, G_+ (x)$ and $C_- (x) = \rho \, G_- (x)$ for all $x \in \bbR$, and $C_+(0) = \chi$.
\end{definition}

The next lemma shows that one can construct the \emph{minimal} tight excursion profile induced by $\psi$, and its consistency is solely determined by a weighted integral of $\psi$.

\begin{restatable}{lemma}{lemlintight}
\label{lem:lin:tight_extension}
Given any bounded and integrable $\psi : (0, 1] \to [1, \infty)$, 
there exists $\chi > 0$ and a $(\rho, \chi)$-excursion profile $G_\pm$ induced by $\psi := G_+|_{(0, 1]}$ such that, for any $(\rho, \chi')$-excursion profile $G'_\pm$ induced by $\psi$, we have $\chi \leq \chi'$.
Moreover, if $\rho$ is parameterized by $s \in (0, s_\ast]$ as $\rho = \frac{1 + e^s}{2s}$, then $\chi$ is defined as follows:
\begin{equation}
    \chi
    =
    \frac{1}{1+e^s}
    \int_0^1
    \left(e^{2s(1-x)}-1\right)\psi(x)\dd x,
    \label{eq:lin:tight_extension_identity}
\end{equation}
where $s_\ast := 1 + W_0(e^{-1})$ (and hence, $\rho \ge \frac{1}{2W_0(e^{-1})}$).
\end{restatable}

Building upon the above relation between $\rho$ and $\chi$ with respect to $\psi$, we derive our main result on learning-augmented linear search.

\begin{restatable}{theorem}{thmlinsearchub}
\label{thm:lin:Kbranch_construction}
Let $s_K \approx 0.5878$ denote the unique positive solution of $e^s(e^{2s} - 1 + 2s e^s) = (1 + e^s)^2$, and define $K : (0, s_\ast] \to (0, e^{2s_\ast}]$ as
\[
    K(s) := \begin{cases}
        \frac{e^s(e^{2s}-1+2se^s)}{(1+e^s)^2}, & \text{ for $s \in (0, s_K]$;} \\
        \xi(s), & \text{ for $s \in (s_K, s_\ast]$,}
    \end{cases}
\]
where $\xi(s)$ is the unique solution $\xi \in [1, e^{2s_\ast}]$ of $(e^s-\xi)\ln \xi + \xi (3+e^{-s})=e^s(e^s+2s-1)$.
Then, for every trade-off parameter $s \in (0, s_\ast]$, there exists a \((\rho(s),\chi(s))\)-excursion profile with
\[
\rho(s) = \frac{1+e^s}{2s},\quad 
\chi(s)=
\begin{cases}
\frac{e^{2s}-1-2s}{2s(1+e^s)} = \rho K(s)e^{-s} - 1, & \text{ for $s \in (0,  s_K]$},\\
\frac{e^{2s}+1+(1+K(s))\log K(s)-2K(s)-2s}
{2s(1+e^s)}, &  \text{ for $s \in (s_K,  s_\ast]$}.
\end{cases}
\]
\end{restatable}

\begin{restatable}{remark}{rmklinasymp}
    The profile of \Cref{thm:lin:Kbranch_construction} yields a $\roblin$-robust $\conslin$-consistent strategy with
    \(
    \roblin = \frac{2}{\conslin - 1} + \frac{7}{3} + O(\conslin - 1).
    \)
    In Appendix~\ref{sec:linsearch_lb}, we present a lower bound of
    \(
    \roblin = \frac{2}{\conslin - 1} + \frac{3}{2} + \Omega(\conslin - 1),
    \)
    demonstrating our strategy is near-optimal up to sub-leading constant.
    Both trade-off bounds are depicted in \Cref{subfig:linear-search-consistency-robustness-search}.
\end{restatable}

%% file: arxiv_include/5_conclusion.tex
\section{Conclusion}
We study learning-augmented online bidding and completely characterize the robustness--consistency trade-off for randomized algorithms via bidding profiles.  
We further extend our approach to the linear search problem, demonstrating its versatility.  

Several interesting directions arise from this work.  A natural next step is to close the remaining gap in the robustness--consistency trade-off for linear search; we conjecture that our profile-based approach leads to a Pareto-optimal strategy. Moreover, since our approach is tailored to the robustness-consistency trade-off, it would be interesting to investigate the optimal \emph{smoothness} achievable for these problems. Finally, it would also be interesting to find other problems in which our profile-based approach can be effectively applied.

%% file: arxiv_include/A_relwork.tex
\section{Further related work} \label{app:relwork}
Online bidding and linear search (also known as the cow-path problem) are classical problems in online optimization~\cite{chrobak2008incremental,bellman1963optimal,beck1964linear,beck1970yet,demaine2006online,kranakis2024survey}.  
Techniques developed for these problems have played a fundamental role in the design and analysis of online algorithms, with applications including contract scheduling, incremental clustering, hierarchical clustering, and list scheduling.
They have also influenced the study of offline optimization problems; for example, they have been used to establish connections between the minimum latency problem and $k$-TSP.  
We refer the interested reader to the survey of Chrobak and Kenyon-Mathieu~\cite{chrobak2006sigact} and the references therein.

We briefly review previous work on learning-augmented algorithms for linear search.  
Angelopoulos~\cite{angelopoulos2023online} proposed a deterministic Pareto-optimal learning-augmented algorithm for the problem.  
Later, Angelopoulos and Simon~\cite{angelopoulos2025learning} studied randomized learning-augmented algorithms for linear search alongside their results for online bidding.

Recent work on learning-augmented algorithms has also shown growing interest in \emph{distributional predictions}~\cite{dinitz2024binary,bai2025optimal,canonne2025little,kehne2025prophet,besbes2025beyond,cui2026ski,kim2026robust}.  
In this direction, online bidding and linear search with distributional predictions have been studied in~\cite{angelopoulos2024contract,angelopoulos2025learning}.

Another important aspect of learning-augmented algorithms is \emph{smoothness}, which relates algorithmic performance to prediction error~\cite{mitzenmacher2022algorithms}.  
A substantial body of work has focused on designing algorithms whose performance degrades smoothly as the prediction error increases~\cite{im2023online,angelopoulos2024contract,elenter2024overcoming}.  
In contrast, this paper focuses on characterizing the trade-off between robustness and consistency, leaving the question of optimal smoothness for these problems as an interesting direction for future work.

For a comprehensive collection of work on learning-augmented algorithms (or algorithms with predictions), we refer the reader to the following repository.\footnote{Alexander Lindermayr and Nicole Megow. \url{https://algorithms-with-predictions.github.io/}, Accessed: 2026-05-05.}

%% file: arxiv_include/B_moski_ob.tex
\newcommand{\Tski}{T_\mathrm{ski}}
\newcommand{\Pski}{P_\mathrm{ski}}
\newcommand{\optski}{\mathrm{OPT}_\mathrm{ski}}
\newcommand{\Bski}{B_\mathrm{ski}}
\section{Equivalence between multi-option ski rental and online  bidding} \label{app:moskiob}
It suffices to show a reduction from multi-option ski rental to online bidding with performance guarantees preserved, because the other direction has already been shown by \cite{shin25}.\footnote{In fact, Shin et al.~\cite{shin25} established an (almost) factor-preserving reduction from the \emph{button} problem, a finite variant of online bidding in which the minimum and maximum values are given in advance.  
One can then derive \Cref{lem:moski:ski2bid} by letting the maximum value tend to $\infty$.}
\begin{lemma}[cf. Lemma 5.1 of \cite{shin25}] \label{lem:moski:ski2bid}
    If there exists a $\rho$-robust, $\chi$-consistent algorithm for multi-option ski rental, then for all constant $\varepsilon \in (0, 1)$, there exists a $(\rho + \varepsilon)$-robust, $(\chi + \varepsilon)$-consistent algorithm for online bidding.
\end{lemma}

In the multi-option ski rental problem, we are given a set of renting options, each of which is defined by a tuple $(c, d) \in \bbR_{\geq 0} \times \left( \bbZ_{\geq 1} \cup \{\infty\}\right)$, representing that we are allowed to go skiing for $d$ days at cost $c$. 
The objective is to cover at the minimum total cost the whole skiing days $\Tski \in \bbZ_{\geq 1}$, which is unknown until the last day.
For the learning-augmented setting, we are additionally given a prediction $\Pski \in \bbZ_{\geq 1}$ on the last day $\Tski$.
For each $t \in \bbZ_{\geq 1}$, let $\optski(t)$ denote the minimum cost of a multiset of renting options to span at least $t$ day.
For any $x \in \bbR_{\geq 0}$, let $\Bski(x)$ denote a multiset of renting options that span the most number of days within the budget $x$; in other words, $\Bski(x)$ covers $\max \{t \in \bbZ_{\geq 1} : \optski(t) \leq x \}$ days.

We are now ready to describe the reduction of multi-option ski rental to online bidding.
Let $\cA$ denote a learning-augmented algorithm for online bidding.
Given a multi-option ski rental instance (except unknown $\Tski$), we first feed to $\cA$ the prediction $P := \optski(\Pski)$.
We then follow $\cA$ by, if $\cA$ makes a bid of $x$, appending $\Bski(x)$ into our solution for multi-option ski rental.
Whenever we are informed of the last day $\Tski$, we terminate the whole execution with reporting to $\cA$ that the target bid $T$ is $x'+\varepsilon$, where $x'$ is the second-to-last bid made by $\cA$, and $\varepsilon$ is infinitesimal.

Note that $x' < \optski(\Tski)$ since, otherwise, we would already have covered at least $\Tski$ days when $\cA$ made the bid $x'$ due to the definition of $\Bski(\cdot)$, contradicting that $x'$ is the second-to-last bid.
Hence, if the sum of all bids made by $\cA$ until its termination is within a factor $\rho$ of $T = x' + \varepsilon \leq \optski(\Tski)$, we can derive that the total cost of our solution for multi-option ski rental is also within the same factor $\rho$ of $\optski(\Tski)$, showing that the robustness follows.
For consistency, it suffices to consider the case where we report to $\cA$ that the target bid is $\optski(\Pski)$, where the rest of the argument follows analogously.
These together imply the following lemma.

\begin{lemma}
    If there exists a $\rho$-robust, $\chi$-consistent algorithm for online bidding, there exists a $\rho$-robust, $\chi$-consistent algorithm for multi-option ski rental.
\end{lemma}

%% file: arxiv_include/C_bid.tex
\section{Additional details on online bidding}

\input{arxiv_include/C1_bid_profile}
\input{arxiv_include/C2_bid_tight}
\input{arxiv_include/C3_bid_str2prof}

%% file: arxiv_include/C1_bid_profile.tex
\subsection{Deferred proofs from \Cref{subsec:profile_function_framework}}
\rmkprofile*
\begin{proof}
    If $G$ is zero at some point, take supremum point $x_0$ for the zero set. The left-continuity implies that $G(x_0) = 0$, and from the robustness condition,
    \[
    \int_{-\infty}^{x_0 + 1} G(t) \dd t \le \rho G(x_0) = 0,
    \]
    implying that $G(t) \equiv 0$ for $(-\infty, x_0 + 1)$, contradicting the maximality of $x_0$.
\end{proof}

\lembidcost*
\begin{proof}
    We first claim that $\tau(T)\leq \rho T-1<\infty$.
    Note that
    \[G(\rho T-1)\geq\frac{1}{\rho}\int_{-\infty}^{\rho T}G(t)\,dt\geq \frac{1}{\rho}\int_{0}^{\rho T }1\,dt=T,\]
    where the first inequality follows from the robustness condition and the second inequality from the offset condition.
    Since $G$ is non-decreasing, $G(t)\geq T$ if $t\geq \rho T-1$, hence $\tau(T)\leq \rho T-1$.

    Next, we analyze the expected cost $\cost_G(T)$.
    After $U\sim\Unif(0,1]$ is sampled, note that for any $n\in\mathbb{Z}$, the algorithm makes the bid $G(n+U)$ if and only if $G(n-1+U)<T$.
    Therefore,
    \begin{align}
        \cost_G(T)
        &=\mathbb{E}_{U\sim \Unif(0,1]}{\sum_{n\in\bbZ} G(n + U) \mathbf{1}_{G(n - 1 + U)<T}}\nonumber\\
        &= \sum_{n\in\bbZ}\int_0^1 G(n+u) \mathbf{1}_{G(n - 1 + u)<T} \, \dd u \nonumber
        \\&
        = \int_{\{t : G(t-1) < T\}} G(t)\dd t \nonumber
        \\&
        = \int_{-\infty}^{\tau(T)+1}G(t)\dd t , \notag 
    \end{align}
    which concludes the proof.
\end{proof}

%% file: arxiv_include/C2_bid_tight.tex
\subsection{Deferred proofs in Section~\ref{subsec:bid:tight}}
\label{app:bid:tight}
Recall that
\[
    \rho(s):=\frac{e^s}{s},
    \qquad
    \chi(s):=
    \begin{cases}
    (e^s-1)/s & (0<s\le\ln2),\\[1mm]
    \xi(s)/s & (\ln2\le s\le1),
    \end{cases}
\]
where $\xi(s)$ is the unique solution $\xi \in[1,e]$ satisfying $\xi(2-\ln\xi)=e^s$.

We show the lemma that will be useful in the subsequent proof.
\begin{lemma}
    Fix a positive number $0<s\le 1$. For some $\chi$, let $G$ be a tight $(\rho(s),\chi)$-bidding profile induced by an integrable function $\phi:(0, 1] \to [1, \infty)$. 
    Then
    \begin{align}
        \chi \ge e^s\int_0^1e^{-sx}\phi(x)\dd x. \notag
    \end{align}
    \label{lem:bidding_bpb}
\end{lemma}
\begin{proof} 
    Let us use $\rho$ to denote $\rho(s)$ for simplicity.
    For $x \in \bbR$, let
    \[
    A(x) := \int_{-\infty}^{x} G(t) \dd t.
    \]
    Note that $A'(x) = G(x)$ a.e. (almost everywhere) as $G$ is non-decreasing.  
    From the ``tightness'' of a tight bidding profile $G$ for $x\le 0$, we obtain a delayed-differential equation (which holds in equality):
    \[
    \rho A'(x) = A(x+1), \qquad x \le 0, \text{ a.e.}
    \]
    Let us assume that $e^{-sx}A(x)$ is non-decreasing.
    Then the following is well-defined:
    \[
        L := \lim_{t \to -\infty} e^{-st}A(t) \ge 0.
    \]
    Note that, for $x\le 0$,
    \[
    \int_{-\infty}^{0} e^{-sx}(\rho A'(x) - A(x+1)) \dd x = 0.
    \]
    Integrating by parts,
    \begin{align*}
    \int_{-\infty}^{0} e^{-sx}A(x+1) \dd x &= e^{s} \int_{-\infty}^{1} e^{-sx}A(x) \dd x \\
    &= e^{s}\left(\left.-\frac{e^{-sx}A(x)}{s}\right\rvert_{-\infty}^{1} + \int_{-\infty}^{1} \frac{e^{-sx}A'(x)}{s}\dd x\right)\\
    &= -\frac{\chi}{s} + \rho L + \rho \int_{-\infty}^{1} e^{-sx}A'(x)\dd x,
    \end{align*}
    where we recall $\rho = e^s/s$ for the last equality.
    We thus obtain
    \begin{align*}
    \int_{-\infty}^{0} e^{-sx}(\rho A'(x) - A(x+1))\dd x &= \frac{\chi}{s} - \rho L - \rho \int_{0}^{1} e^{-sx}A'(x) \dd x.
    \end{align*}
    Since $A'(x) = \phi(x)$ for $0 < x \le 1$ a.e., and $\rho = \rho(s)= e^{s}/s$,
    \[
    \chi = e^{s}L + e^{s}\int_{0}^{1} e^{-sx}\phi(x) \dd x \ge e^{s} \int_{0}^{1} e^{-sx}\phi(x) \dd x,
    \]
    where the inequality is from $L\ge 0$.

    It remains to show $e^{-sx}A(x)$ is non-decreasing. 
    Since $A(x)$ is continuous, 
    \begin{align}
    e^{-sx}A(x) \text{ is non-decreasing  on } \bbR \iff A'(x) \ge sA(x) \text{ a.e.} \label{eq:bid:dde_laplace}
    \end{align}
    
    To this end, we show that $A'(x) \ge s A(x)$ a.e. First observe that, 
    \begin{align*}
        \rho A'(x) &\ge A(x+1) \\
        &= A(x) + \int_{x}^{x+1} A'(t) \dd t\\
        &\ge A(x) + \int_{x}^{x+1} A'(x) \dd t\\
        &= A(x) + A'(x),
    \end{align*}
    where the inequality is due to the fact that $A'(t)=G(t)$ is non-decreasing. 
    Thus,
    \[
    A'(x) \ge \frac{1}{\rho-1}A(x).
    \]
    Now define $c_{0} := \frac{1}{\rho-1}$, and $c_{n} := e^{c_{n-1}}/\rho$ for $n \in \bbZ_{\ge 1}$. 
    We show $A'(x)\ge c_n A(x)$.
    From definition, $A'(x) \ge c_{0} A(x)$.
    Let us assume $A'(x) \ge c_{n-1}A(x)$.
    Since 
    \[
    \ln A(x+1) - \ln A(x) = \int_{x}^{x+1} \frac{A'(t)}{A(t)}\dd t \ge c_{n-1},
    \]
    we have that \[A'(x) \ge \frac{1}{\rho}A(x+1) \ge \frac{1}{\rho}e^{c_{n-1}} A(x) = c_{n}A(x).\]
    We now show $\{c_n\}$ converges to $s$ as $n\to\infty$.
    We first show that it is increasing and bounded from above by $s$.
    Note that $c_{0} < s$. 
    Suppose $c_{n-1} < s$.
    It is easy to see that $c_{n} = e^{c_{n-1}}e^{-s}s < s$. 
    Since $e^z/z$ strictly decreases on $z\in(0,1]$,
    \[
    \frac{c_n}{c_{n-1}} = \frac{e^{c_{n-1}}/\rho}{c_{n-1}} = f(c_{n-1})/f(s) > 1,
    \]
    implying $c_{n} > c_{n-1}$.
    By the monotone convergence theorem, $c_{n} \to c^{\ast}$. Since $e^{c^{\ast}} = \rho c^{\ast}$, $c^{\ast} = s$. Therefore, $A'(x) \ge c^{\ast}A(x) = sA(x)$, showing \eqref{eq:bid:dde_laplace}.
\end{proof}

We show a useful fact.
\begin{fact}\label{fact:schi}
    Given $s\in (0,1]$, $s > \ln 2 \iff s\chi(s) > 1$.
    Moreover, if $s\chi(s)>1$, $0\le x_0<1$ where $x_0:=1-\frac{\ln (s\chi(s))}{s}$.
\end{fact}
\begin{proof}
    We first prove the first equivalence, beginning with the if part.
    Suppose toward contradiction that $s \le \ln 2$.
    Then $s\chi(s)=e^s-1$.
    On the other hand, $e^s-1 > 1\iff s > \ln 2$.
    
    For the only if part, fix any $s>\ln 2$.
    From $e^s>2$ and the fact that $x \mapsto x(2-\ln x)$ strictly increases on $[1,e]$, $\xi(s)(2-\ln\xi(s))=e^s>2=1\cdot (2-\ln 1)$, we obtain $\xi(s)>1$.
    Observe $s\chi(s)=\xi(s)$ from the definition of $\chi(s)$.

    For the second statement, $x_0 < 1$ follows from $s\chi(s) > 1$.
    Note that $s\chi(s)=\xi(s)$ since $s>\ln 2$.
    Furthermore, $2- \ln \xi(s) \ge 1$ since $\xi(s)\in[1,e]$.
    Therefore,
    \[
        s\chi(s) = \xi(s) \le \xi(s) (2-\ln \xi(s)) = e^s
    \]
    The following equivalence completes the proof.
    $x_0 \ge 0 \iff s\chi(s) \le e^s.$
\end{proof}

\subsubsection{Optimal bidding profile}
\thmbiddingub*

Observe that $\xi(s)$ is well defined since $\xi \mapsto \xi (2 - \ln \xi)$ is increasing on $\xi \in [1, e]$ with $1 \mapsto 2$ and $e \mapsto e$.

\begin{proof}
    It is trivial for $s=1$; the classical $(e, e)$-bidding profile $G(x)=e^x$ suffices.

    Now fix any $0<s<1$. 
    Let us use $\rho$ and $\chi$ to denote $\rho(s)$ and $\chi(s)$, respectively.
    Throughout the proof let $\phi(x) := \max(1,s\chi \cdot e^{s(x-1)})$ be a function defined on $(0,1]$.
    
    We first define $G$ when $x > 0$:
    \begin{align*}
        G(x):=\begin{cases}
            \phi(x),&0<x\le 1,\\
            \max(1,s\chi)\cdot e^{s(x-1)},& x>1.
        \end{cases}
    \end{align*}
    It remains to define $G$ on $x\le0$; we inductively define the remaining $G$.
    Let $G(0)=\frac{\chi}{\rho}$.
    Given $G$ defined on $[k,k+1)$ for an integer $k \le 0$, we define $G$ on $[k-1,k)$ by
    \[
        G(x)=G(k)-\frac{1}{\rho}\int_{x+1}^{k+1}G(t)\dd t.
    \]
    Note that $G$ is continuous on $x\le 0$.
    
    Suppose $G$ is non-negative and non-decreasing.
    We show that $G$ is a $(\rho,\chi)$-bidding profile.
    The offset condition trivially holds.
    From the definition, for any integer $k\le 0$ and $x \in [k-1,k)$,
    \begin{align}
        G(x) \ge G(x)-G(k) + \sum_{n\le k}(G(n)-G(n-1)) = \frac{1}{\rho}\int_{-\infty}^{x+1}G(t)\dd t,\label{eq:bid:ub-1}
    \end{align}
    which shows the robustness when $x< 0$.
    The above also holds when $x=0$ by the choice of $G(0)$, showing the consistency condition and the robustness condition when $x=0$.
    
    We show the robustness condition when $x>0$.
    Suppose  $s\chi \le 1$. From \Cref{fact:schi}, $s \le \ln 2$.
    Then,
    \[
        \chi=\frac{e^s-1}{s},
        \qquad
        G(x)=1\quad(0<x\le1),
        \qquad
        G(x)=e^{s(x-1)}\quad(x\ge1).
    \]
    Note that
    \[
        \int_{-\infty}^{x+1}G(t)\dd t\le \chi+\frac{e^{sx}-1}{s}.
    \]
    If $0<x\le1$, then
    \[
        \chi+\frac{e^{sx}-1}{s}
        =
        \frac{e^s+e^{sx}-2}{s}
        \le
        \frac{e^s}{s}
        =
        \rho G(x),
    \]
    because $s\le\ln2$ implies $e^{sx}\le2$. If $x\ge1$, then
    \[
        \chi+\frac{e^{sx}-1}{s}
        \le
        \frac{e^{sx}}{s}
        =
        \rho e^{s(x-1)}
        =
        \rho G(x),
    \]
    because $\chi\le1/s$.
    
    Now suppose $s\chi>1$. Then $G(x)=s\chi\cdot e^{s(x-1)}$ for $x \ge 1$. Hence
    \[
        \int_{-\infty}^{x+1}G(t)\dd t \le\chi+\int_1^{x+1}s\chi\cdot e^{s(t-1)}\dd t
        =
        \chi e^{sx}.
    \]
    It is easy to see that the robustness condition holds when $x> 1$.
    Consider any $x\in (0,1]$.
    If $G(x)=s\chi \cdot e^{s(x-1)}$, then
    \[
        \rho G(x)=\frac{e^s}{s}\cdot s\chi \cdot e^{s(x-1)}
        =
        \chi e^{sx}.
    \]
    If $G(x)=1$, which implies $s\chi \cdot e^{s(x-1)}\le1$, then the above equation shows $\chi e^{sx} \le \rho = \rho G(x)$.

    Note that \eqref{eq:bid:ub-1} holds in equality.
    Suppose toward contradiction that $\lim_{n \to -\infty}G(n)=\varepsilon$ for some $\varepsilon>0$.
    Then, for any integer $N<0$, $G(0)-G(N)\ge \frac{|N|\varepsilon}{\rho}$ since $G$ is non-decreasing.
    Sufficiently small $N$ leads to $G(N)<0$, contradicting the non-negativity of $G$.
    Therefore, $G$ is a tight $(\rho,\chi)$-bidding profile.

    In what follows, we show $G$ is non-negative and non-decreasing.
    Recall the definition of the functional operator $\cF$.
    \defncF*
    Observe that our bidding profile $G$ on $x\le 0$ can be equivalently represented as the $(-\infty,0]$ region of a non-negative, non-decreasing function $G^\ast$ defined on $(-\infty,1]$ such that $G^\ast = \cF G^\ast$, $\int_{-\infty}^1 G^\ast(t)\dd t = \int_{-\infty}^1 G(t)\dd t$ and $G^\ast\equiv \phi$ on $(0,1]$.
    This can be seen from $G^\ast(0)=G(0)$, and for $x \in [k-1,k)$ for some integer $k\le 0$,
    \begin{align*}
        G^\ast(x) = \frac{1}{\rho} \int_{-\infty}^{x+1} G^\ast(t) \dd t &= \frac{1}{\rho} \left( \int_{-\infty}^{k+1} G^\ast(t) \dd t - \int_{x+1}^{k+1} G^\ast(t) \dd t \right)
        \\&
        = G^\ast(k) - \frac{1}{\rho} \int_{x+1}^{k+1} G^\ast(t) \dd t = G(x)
    \end{align*}
    where the first and the third equalities are from $G^\ast = \cF G^\ast$ and the last equality is from the inductive hypothesis.
    
    We show such $G^\ast$ exists.
    We exploit the \emph{order-preserving} property of $\cF$: if $A \le B$ (i.e., $A(x)\le B(x)$ for all $x$ in the domain) for two functions $A,B$, then $\cF A \le \cF B$.
    Let $G^{(0)}$ be a function defined on $x\le 1$ such that $G^{(0)}(x)=0$ for $x\le 0$ and $G^{(0)}(x)=\phi(x)$ for $0<x\le 1$.
    For $n \in \bbZ_{\ge 0}$, $G^{(n+1)}(x) := \cF G^{(n)}(x)$.
    Note that the order-preserving property shows the following implication:
    \[
        G^{(n)}\le G^{(n+1)} \Longrightarrow \cF G^{(n)}\le \cF G^{(n+1)} \Longrightarrow G^{(n+1)}\le G^{(n+2)}.
    \]
    Since $G^{(0)} \le G^{(1)}$, we have $G^{(n)}\le G^{(n+1)}$ for all $n \in \bbZ_{\ge 0}$.
    
    We show for all $n \in \bbZ_{\ge 0}$, $G^{(n)}$ is bounded from above by 
    \begin{align*}
        H(x):=\begin{cases}
            Ke^{cx}, & x\le 0, \\
            \phi(x), &0<x\le 1,
        \end{cases}
    \end{align*}
    where $c\in (s,1)$ and $K>0$ is sufficiently large.
    
    First we show $\cF H \le H$.
    Since $x \mapsto e^x/x$ strictly decreases on $0<x\le 1$, $e^c/c < e^s/s=\rho$.
    Therefore, for $x\le -1$,
    \[
        \cF H(x) = \frac{1}{\rho}\int_{-\infty}^{x+1}H(t)\dd t = \frac{1}{\rho}Ke^{cx}\frac{e^c}{c} \le Ke^{cx} = H(x).
    \]
    For $x \in (-1,0]$, using $\phi(x)\le M e^{sx}$ for some $M>0$,
    \begin{align*}
        \cF H(x) &= \frac{1}{\rho}\left(\int_{-\infty}^{0}Ke^{cx}\dd t +\int_{0}^{x+1}\phi(t)\dd t\right) 
        \\&
        \le \frac{1}{\rho}\left(\frac{K}{c} + \int_{-\infty}^{x+1}M e^{st} \dd t\right)
        \\&
        \le \frac{K}{\rho c} + \frac{Me^s}{\rho s}e^{sx} \le Ke^{cx} = H(x).
    \end{align*}
    where the last inequality is from $c>s$ and the choice of sufficiently large $K$. 
    Along with $G^{(0)}\le H$ (which holds from the definition), we obtain the following:
    \[
    G^{(n)} \le \cF^{(n)} G^{(0)} \le \cF^{(n)} H \le H.
    \]
    By the monotone convergence theorem, $\{G^{(n)}\}$ converges.
    We show $\lim_{n\to \infty}G^{(n)}=G^\ast$.
    From the definition of $\cF$, it is easy to see that $G^\ast$ is non-negative and non-decreasing and $G^\ast \equiv \phi$ on $(0,1]$.
    Since $G$ is a tight $(\rho,\chi)$-bidding profile, showing the following completes the proof.
    \[
        \chi^\ast := \int_{-\infty}^1 G^\ast(t)\dd t = \chi.
    \]
    Consider a $(\rho,\chi^\ast)$-bidding profile defined by $G^*(x)$ on $x\le 1$ and $G(x)$ on $x>1$.
    Observe $G^\ast = \cF G^\ast$ implies $G^\ast(x)=\frac{1}{\rho}\int_{-\infty}^{x+1}G^\ast(t)\dd t$.
    Therefore, such bidding profile is tight and induced by $\phi$.
    By applying \Cref{lem:bidding_bpb}, 
    \[
        \chi^\ast \ge e^s \int_0^1 e^{-sx}\phi(x)\dd x.
    \]
    In fact, in the proof of \Cref{lem:bidding_bpb}, the only inequality was from 
    \[
    L:=\lim_{x\to -\infty}e^{-sx}\int_{-\infty}^x G^\ast(t) \dd t \ge 0.
    \]
    Recall $G^\ast(x)$ is bounded from above by $Ke^{cx}$ on $x\le 0$ where $c > s$ and thus $L=0$, implying
    \[
        \chi^\ast = e^s \int_0^1 e^{-sx}\phi(x)\dd x.
    \]
    Suppose $s\chi \le 1$. From \Cref{fact:schi}, $s\le \ln 2$, which implies $\phi(x)=1$ for all $0<x\le 1$, and thus $\chi^\ast=\chi$.
    Otherwise, let $x_0:=1-\ln(s\chi)/s \in [0,1)$. 
    Then
    \begin{align*}    
        \chi^\ast &= e^s \left( \int_0^{x_0} e^{-sx}\cdot 1\dd x + \int_{x_0}^1 e^{-sx}s\chi\cdot e^{s(x-1)}\dd x \right)
        \\&
        =\frac{e^{s}-s\chi}{s} + \chi\ln(s\chi)
        \\&
        =\frac{\xi(2-\ln\xi)-s\chi}{s} + \chi\ln(s\chi)
        \\&
        =\frac{s\chi(2-\ln(s\chi))-s\chi}{s} + \chi\ln(s\chi)
        \\&
        = \chi.
    \end{align*}
\end{proof}

\subsubsection{Lower bounds}
\lembidtightening*
\begin{proof}
    Recall the functional operator $\cF$ defined in the proof of \Cref{thm:bidding_ub}.
    In a high level, we show applying $\cF$ to $G$ yields a $(\rho,\chi')$-bidding profile with $\chi'\le \chi$.
    We then show that by iteratively applying $\cF$ to $G$, the process converges to a tight bidding profile.
    Let $G^{(0)} := \left. G \right\rvert_{(-\infty, 1]}$, and for $n \in \bbZ_{\ge 0}$, $G^{(n+1)}(x) := \cF G^{(n)}(x)$.
    Note that the order-preserving property of $\cF$ shows the following implication:
    \[
        G^{(n)}\ge G^{(n+1)} \Longrightarrow \cF G^{(n)}\ge \cF G^{(n+1)} \Longrightarrow G^{(n+1)}\ge G^{(n+2)}.
    \]
    From the robustness condition of bidding profiles and $\cF$,
    for $x \le 0$,
    \[
        G^{(0)}(x) \ge \frac{1}{\rho}\int_{-\infty}^{x+1}G^{(0)}(t)\dd t = \cF G^{(0)}(x) = G^{(1)}(x).
    \]
    Since $G^{(0)}(x) = \cF G^{(0)}(x)$ for $x \in (0,1]$, $G^{(0)} \ge G^{(1)}$ and thus we have $G^{(n)}\ge G^{(n+1)}$ for all $n \in \bbZ_{\ge 0}$.
    Since applying $\cF$ to a non-negative function always yields a non-negative function, $\{G^{(n)}\}$ converges to a function, say $G^{\ast}$, as $n \to \infty$ by the monotone convergence theorem.

    To see that $G^\ast$ glued with $G$ for $x>1$ is a tight $(\rho,\chi')$-bidding profile for some $\chi'\le \chi$, we simply follow the proof of \Cref{thm:bidding_ub}.
    In particular, 
    \[
    \chi' := \int_{-\infty}^{1} G^{\ast}(t)\dd t \le \int_{-\infty}^{1} G(t) \dd t \le \chi.
    \]
    This completes the proof.
\end{proof}

We show the following lemma before proving \Cref{thm:bidding_lb}.
\begin{lemma}\label{lem:bidding_phi}
    Fix $0<s \le 1$, and an integrable function $\phi : (0, 1] \to [1, \infty)$. If there exists a tight $\left(\rho:=\rho(s), \chi\right)$-bidding profile induced by some $\phi$, then
    \begin{align*}
        \phi(x)\ge \max(1,s\chi e^{s(x-1)}).
    \end{align*}
\end{lemma}
\begin{proof}
    As in the proof of \Cref{lem:bidding_bpb}, define $A(x) := \int_{-\infty}^{x} G(t)\dd t.$
    Recall from \eqref{eq:bid:dde_laplace} in the proof of \Cref{lem:bidding_bpb} that $A'(x) \ge s A(x)$ for $x \in \bbR$.
    Therefore, for $x>0$,
    \[
    \ln A(x+1) - \ln \chi = \ln A(x+1) - \ln A(1) = \int_{1}^{x+1} \frac{A'(t)}{A(t)}\dd t \ge \int_{1}^{x+1} s\dd t = sx.
    \]
    From the robustness condition for $0 < x \le 1$,
    \[
    \phi(x) = A'(x) \ge \frac{1}{\rho}A(x+1) \ge \frac{1}{\rho}A(1)e^{sx} = s\chi e^{s(x-1)}.
    \]
    Meanwhile, we have $\phi(x) \ge 1$ from the definition of a bidding profile, completing the proof.
\end{proof}

We are now ready to prove the following theorem.
\thmbiddinglb*
\begin{proof} 
    Let $G$ be a tight $(\rho(s),\chi)$-bidding profile, and let $\phi$ be a function on $(0,1]$ where $G$ is induced.
    From \Cref{lem:bidding_bpb,lem:bidding_phi},
    \begin{align}
        \chi \ge e^{s}\int_{0}^{1}e^{-sx}\phi(x) \dd x \ge e^{s}\int_{0}^{1} e^{-sx}\max(1, s\chi e^{s(x-1)})\dd x.\label{eq:bid:chilb}
    \end{align}
    Note that the following holds in general. 
    \[
    \chi \ge e^s\int_0^1 e^{-sx}\dd x = \frac{e^s-1}{s}.
    \]
    Suppose toward contradiction that $\chi < \chi(s)$.
    If $s \le \ln 2$, then $\chi(s)=(e^s-1)/s$ which contradicts $\chi \ge (e^s-1)/s$.
    Thus, it must be the case that $s > \ln 2$.
    Note $s \chi \ge e^s-1 > 1$.
    
    From \Cref{fact:schi}, $s\chi(s)>1$ and $x_0 \in [0,1)$ where 
    \[
        x_{0}:=1 - \frac{\ln(s\chi(s))}{s}.
    \]
    Let 
    \[
        y_{0}:=1 - \frac{\ln(s\chi)}{s}.
    \]
    Note that $s\chi > 1$ implies $y_0 < 1$.
    Since $s\chi<s\chi(s)$, we have $y_0 > x_0$.
    Thus, $y_0 \in (x_0,1)\subseteq(0,1)$.

    By further bounding \eqref{eq:bid:chilb} from below,
    \[
    \chi \ge e^{s} \left(\int_{0}^{y_0} e^{-sx}\dd x + \int_{y_0}^{1} e^{-sx}s\chi e^{s(x-1)}\dd x \right) = \frac{e^{s}-s\chi}{s} + \chi\ln(s\chi),
    \]
    implying
    \[
    s\chi(2 - \ln(s\chi)) \ge e^{s}.
    \]
    Let $f(x):=x(2-\ln x)$.
    Recall that $f(x)$ is strictly increasing on $[1,e]$.
    By definition, $\xi(s)=s\chi(s) \in [1,e]$ satisfies $f(s\chi(s))=e^s$.
    Recall $1<s\chi<s\chi(s)\le e$.
    The strict monotonicity of $f$ implies 
    \[
        f(s\chi)<f(s\chi(s))=e^s,
    \]
    which directly contradicts $f(s\chi)=s\chi(2 - \ln(s\chi)) \ge e^{s}$.
\end{proof}

%% file: arxiv_include/C3_bid_str2prof.tex
\subsection{Deferred proofs from \Cref{subsec:bid:alg2profile}} \label{app:str2prof}

\lembiddinggconvert*

\begin{proof}
    The if part is proved in \Cref{cor:G_bidding}. We prove the only-if part.

    Let any $\rho$-robust $\chi$-consistent (randomized) strategy $\scriptA$ be given, whose prediction is given as $P=1$.
    Recall $\scriptA$ is a probability distribution over bi-infinite sequences; let $(Q_{n})_{n\in \bbZ}$ denote a random sequence.
    Since the strategy is $\rho$-robust and $\chi$-consistent,
    \begin{equation}
        \cost_\scriptA(T)=\bbE{\sum_{n\in\bbZ}Q_n\mathbf 1_{Q_{n-1}<T}}\le \rho T, \quad \forall T>0,\label{eq:bid_alg_rob}
    \end{equation}
    and
    \begin{equation}
        \cost_\scriptA(1)=\bbE{\sum_{n\le0}Q_n}\le \chi.\label{eq:bid_alg_cons}
    \end{equation}
    
    Let $\mu_n$ be the marginal probability distribution of $Q_n$, and $\mu$ be its aggregate measure.
    \[
        \mu:=\sum_{n\in\bbZ}\mu_n.
    \]
    In other words, for every bounded interval $I\subseteq(0,\infty)$,
    \[
        \mu(I)=\sum_{n\in\bbZ}\bbE{\mathbf{1}_{Q_n\in I}}
        =
        \bbE{\sum_{n\in\bbZ}\mathbf{1}_{Q_n\in I}}.
    \]

    Clearly, $\mu$ is supported on $(0, \infty)$.
    {
    Now we verify that $\mu$ is locally finite on $(0,\infty)$ 
    while $\mu((0,1))=\mu([1,\infty))=\infty$.
    For any $T\in(0,\infty)$, consider a set $[T,(2\rho-1)T)$ and let $x:=\mu([T,(2\rho-1)T))$.
    {Since $\cost_\scriptA((2\rho-1)T)\le \rho (2\rho-1)T$}, we have
    \begin{align*}
        Tx&=\bbE{\sum_{n\in \bbZ}T\cdot \mathbf{1}_{Q_n\in [T,(2\rho-1)T)}}\\
        &\leq \bbE{\sum_{n\in \bbZ}Q_n\cdot \mathbf{1}_{Q_n\in [T,(2\rho-1)T)}}\\
        &\leq \bbE{\sum_{n\in \bbZ}Q_n\cdot \mathbf{1}_{Q_{n-1}<(2\rho-1)T}}\leq \rho(2\rho -1)T,
    \end{align*}
    thus $x\leq \rho(2\rho -1)$.
    On the other hand, {from $\cost_\scriptA(T)\le \rho T$}, we have
    \begin{align}
     \rho T&\geq \bbE{\sum_{n\in \bbZ}Q_n\cdot \mathbf{1}_{Q_{n-1}<T}}\nonumber\\
        &\geq \bbE{\sum_{n\in \bbZ}Q_n\cdot \mathbf{1}_{Q_{n-1}<T\leq Q_n}}\nonumber\\
        &\geq \bbE{\sum_{n\in \bbZ}Q_n\cdot \mathbf{1}_{Q_{n-1}<T\leq Q_n<(2\rho-1)T}}+\bbE{\sum_{n\in \bbZ}Q_n\cdot \mathbf{1}_{Q_{n-1}<T<(2\rho-1)T\leq Q_n}}\nonumber\\
        &\geq T\cdot\min\{1,x\}+(2\rho -1)T\cdot (1-\min\{1,x\}),\label{eq:mu_lb}
    \end{align}
    To check the inequality \eqref{eq:mu_lb}, note that ${\sum_{n\in \bbZ}Q_{n} \cdot \mathbf{1}_{Q_{n-1}<T\leq Q_n}}$ is a random variable that represents the stopping bid of $\scriptA$ when the target is at $T$. If $x<1$, then the stopping bid is larger than or equal to $(2\rho -1)T$ with probability $\geq 1-x$.
    Therefore,
    \begin{equation*}
        2(\rho -1)\min\{1,x\}\geq (\rho - 1),\qquad x\geq \min\{1,x\}\geq\frac{1}{2}.
    \end{equation*}
    Overall, we have
    \[\frac{1}{2}\leq \mu([T,(2\rho-1)T))\leq \rho(2\rho -1),\]
    where the upper bound implies the local finiteness while the lower bound implies $\mu((0,1))=\mu([1,\infty))=\infty$.
    }
    
    We now define \(G : \bbR \to (0, \infty)\) as follows:
\[
    G(x) := \begin{cases}
        \sup\{0< t<1:\ \mu([t,1)) \ge -x\}, & \text{ if $x < 0$;} \\
        \inf\{t\ge1:\ \mu([1,t])\ge x\}, & \text{ if $x \geq 0$.}
    \end{cases}
\]
Since $\mu([t,1))<\infty$ for all $0<t<1$ by the local finiteness of $\mu$ and $\lim_{t\downarrow 0}\mu([t,1))=\infty$,
$G(x)$ is well-defined, finite, and strictly positive for all $x<0$.
Similarly, since $\mu([1,t])<\infty$ for all $t>1$ and $\lim_{t\to \infty}\mu[1,t]=\infty$, the same properties hold for all $x\geq 0$.

Note that for every bounded interval
\(I\subseteq(0,\infty)\),
\begin{align}
    m(G^{-1}(I))=\mu(I),
    \label{eq:bidding_pushforward}
\end{align}
where \(m\) denotes the Lebesgue measure; In other words, \(\mu\) is a pushforward measure of the Lebesgue measure by \(G\). 

We show that $G$ satisfies all the definition of a $(\rho, \chi)$-bidding profile.
\(G\) is non-decreasing by its definition. We may assume that $G$ is left-continuous; otherwise, we can modify $G$ to be left-continuous without changing the value of its integral. This is justified because $G$ is monotone, implying it has only countably many discontinuities, hence the integration of $G$ over the Lebesgue measure like \eqref{eq:bid_precost} and \eqref{eq:bid_postcost} stays intact after the change. The offset condition is clear by the definition of $G$:
\[
G(x)<1\quad(x<0),
\qquad
G(x)\ge1\quad(x>0).
\]

    Finally, we show that the strategy driven from $G$ has no larger expected cost at any target,
    which naturally implies both robustness and consistency conditions.
    Let any target $T>0$ be given.
    From the original strategy, let us decompose the bidding cost for $T$ into
    \begin{align}
        \cost_\scriptA(T)&=\bbE{\sum_{n\in\bbZ}Q_n\mathbf{1}_{Q_{n-1}<T}}\notag\\
        &=\bbE{\sum_{n\in\bbZ}Q_n\mathbf{1}_{Q_{n}<T}}+\bbE{\sum_{n\in\bbZ}Q_n\mathbf{1}_{Q_{n-1}<T\leq Q_n}}.\label{eq:decompose}
    \end{align}
    The first term of \eqref{eq:decompose} simplifies to
    \begin{align}
        \bbE{\sum_{n\in\bbZ}Q_n\mathbf{1}_{Q_{n}<T}}
        &=\sum_{n\in\bbZ}\bbE{Q_n\mathbf{1}_{Q_{n}<T}}\notag\\
        &=\sum_{n\in\bbZ}\int_{(0,T)} t\,d\mu_n(t)\notag\\
        &=\int_{(0,T)} t\,d\mu(t)=\int_{-\infty}^{\tau(T)}G(t)\,dt,\label{eq:bid_precost}
    \end{align}
    where $\tau(T):=\sup\{t\in\mathbb{R}:G(t)<T\}$ is defined as in \Cref{lem:bid_cost}.
    The last equality is due to the change of variable formula with \eqref{eq:bidding_pushforward}, since $G^{-1}(0,T)=(-\infty,\tau(T))$ or $(-\infty,\tau(T)]$.

    We now consider the second term of \eqref{eq:decompose}.
    Note that the second term is just an expectation of the stopping bid.
    Denote $R$ as a random variable representing the stopping bid from $\scriptA$, and $R^G$ as that from a strategy $\cB$ induced by $G$.
    Note that $R^G$ has a same distribution with $G(\tau(T)+U),\,U\sim\Unif(0,1]$.
    Defining $T^*:=G(\tau(T)+1)$, we have
    \begin{align}
        \bbE{\sum_{n\in\bbZ}Q_n\mathbf{1}_{Q_{n-1}<T\leq Q_n}}&=\bbE{R}=\bbE{R\cdot(\mathbf{1}_{T\leq R<T^*}+\mathbf{1}_{R=T^*}+\mathbf{1}_{R>T^*})}\notag\\
        &=\bbE{(R-T^*)\cdot(\mathbf{1}_{T\leq R<T^*}+\mathbf{1}_{R=T^*}+\mathbf{1}_{R>T^*})}+T^*.\label{eq:trich}
    \end{align}
    To bound the first term, we note that for every interval $I\subseteq [T,T^*)$,
    \begin{equation*}
    \Pr[R\in I]\leq \Pr\left[\exists n,Q_n\in I\right]\leq \sum_{n\in\mathbb{Z}}\Pr[Q_n\in I]=\mu(I)=\Pr[R^G\in I].    
    \end{equation*}
    Since $x-T^*<0$ for $T\leq x<T^*$, the above implies that
    \begin{equation}
        \bbE{(R-T^*)\mathbf{1}_{T\leq R<T^*}}\geq \bbE{(R^G-T^*)\mathbf{1}_{T\leq R^G<T^*}}.\label{eq:rt_1}
    \end{equation} 
    Moreover,
    \begin{align}
        \bbE{(R-T^*)\mathbf{1}_{R=T^*}}&=0=\bbE{(R^G-T^*)\mathbf{1}_{R^G=T^*}},\label{eq:rt_2}\\
        \bbE{(R-T^*)\mathbf{1}_{R>T^*}}&\geq 0 =\bbE{(R^G-T^*)\mathbf{1}_{R^G>T^*}}.\label{eq:rt_3}
    \end{align}
    Combining \eqref{eq:rt_1}, \eqref{eq:rt_2} and \eqref{eq:rt_3} to \eqref{eq:trich} yields
    \begin{align}
        \bbE{R}&=\bbE{(R-T^*)(\mathbf{1}_{T\leq R<T^*}+\mathbf{1}_{R=T^*}+\mathbf{1}_{R>T^*})}+T^*\notag\\
        &\geq \bbE{(R^G-T^*)(\mathbf{1}_{T\leq R^G<T^*}+\mathbf{1}_{R^G=T^*}+\mathbf{1}_{R^G>T^*})}+T^*\notag\\
        &=\bbE{R^G}=\int_{0}^1 G(\tau(T)+t)\,dt=\int_{\tau(T)}^{\tau(T)+1}G(t)\,dt.\label{eq:bid_postcost}
    \end{align}

    Overall, applying \eqref{eq:bid_precost} and \eqref{eq:bid_postcost} to \eqref{eq:decompose}, we have
    \begin{align*}
        \cost_\scriptA (T)&=\bbE{\sum_{n\in\bbZ}Q_n\mathbf{1}_{Q_{n}<T}}+\bbE{\sum_{n\in\bbZ}Q_n\mathbf{1}_{Q_{n-1}<T\leq Q_n}}\\
        &\geq \int_{-\infty}^{\tau (T)} G(t)\,dt+\int_{\tau(T)}^{\tau(T)+1} G(t)\,dt\\
        &= \int_{-\infty}^{\tau (T)+1} G(t)\,dt=\cost_\cB(T).
    \end{align*}
    Combining the above with \eqref{eq:bid_alg_rob} and \eqref{eq:bid_alg_cons}, we have for any $\varepsilon > 0$:
    \begin{equation*}
        \int_{-\infty}^{x+1}G(t)\,dt\le \int_{-\infty}^{\tau(G(x) + \varepsilon)+1}G(t)\,dt=\cost_\cB(G(x) + \varepsilon)\leq \cost_\scriptA (G(x) + \varepsilon)\leq \rho (G(x) + \varepsilon).
    \end{equation*}
    Letting $\varepsilon \to 0$, we get that $\int_{-\infty}^{x+1} G(t)\dd t \le \rho G(x)$ for all $x$. For the consistency, From $\tau(1) = 0$:
    \begin{equation*}
        \int_{-\infty}^{1}G(t)\,dt=\int_{-\infty}^{\tau(1)+1}G(t)\,dt=\cost_\cB(1)\leq\cost_\scriptA(1)\leq \chi.
    \end{equation*}
    Therefore, $G$ is a $(\rho,\chi)$-bidding profile, which completes the proof.
\end{proof}

%% file: arxiv_include/D_lin.tex
\section{Additional details on linear search}

\input{arxiv_include/D1_lin_prelim}
\input{arxiv_include/D2_lin_profile}
\input{arxiv_include/D3_lin_ub}
\input{arxiv_include/D4_lin_lb}

%% file: arxiv_include/D1_lin_prelim.tex
\subsection{Preliminaries on linear search}
\label{app:lin:prelim}

We continue with preliminary background on linear search.

As in online bidding, we introduce a strategy for linear search as a surrogate model for an algorithm.
A deterministic strategy is formally defined as a bi-infinite sequence $y = \{y_n\}_{n \in \bbZ}$ of signed excursion endpoints such that the signs of consecutive elements are alternating (i.e., $\sgn(y_n) \neq \sgn(y_{n+1})$ for all $n \in \bbZ$), and the magnitudes of every other element are increasing (i.e., $|y_n| < |y_{n+2}|$ for all $n \in \bbZ$).
Given a deterministic strategy $y$ and a target $T \in \bbR \setminus \{0\}$, the cost incurred by $y$ given $T$ is defined as
\[
    \textstyle
    \cost_y(T) := |T| + 2 \, \sum_{n < n_\ast} |y_n|,    
\]
where $n_\ast := \min \{ n \in \bbZ : \sgn(y_n) = \sgn(T) \text{ and } |y_n| \geq |T| \}$.
A randomized strategy $Y$ is defined as a probability distribution over deterministic strategies, and the expected cost incurred by $Y$ given $T$ is denoted by $\cost_Y(T)$.

As the universe of $T$ is $\bbR \setminus \{0\}$, we can see that the problem is homogeneous under scaling and negating.
We thus assume that the prediction is given as $P := +1$ by dividing the whole universe by $P$.
Moreover, we can freely shift the indices of a strategy since it is bi-infinite.
We may assign index~$0$ to the first excursion point where $+1 = P$ is passed through for the first time.

A (randomized) strategy $Y$ is called $\rho$-robust if $\cost_Y (T) \leq \rho \, |T|$ for any $T \in \bbR \setminus \{0\}$, and $\chi$-consistent if $\cost_Y(+1) \leq \chi$.
A pair $(\rho, \chi)$ is \emph{Pareto-optimal} if there exists a $\rho$-robust $\chi$-consistent strategy, but no $\rho'$-robust $\chi$-consistent strategy for any $\rho' < \rho$.

%% file: arxiv_include/D2_lin_profile.tex
\subsection{Deferred proofs from \Cref{subsec:excursion_profile}}
\label{app:subsec:excursion_profile}

\corlsbid*
\begin{proof}
Since the cost of a linear search always double-counted if an excursion does not reach the target, and ends with a successful excursion exactly incurring the cost at the target's magnitude, we will see that a $(\rho, \chi)$-excursion profile drives a $(1+2\rho)$-robust, $(1+2\chi)$-consistent linear search strategy.

For $x > 0$, define
\begin{align*}
&\tau_{+}(T) := \sup\{t \in \bbR : G_{+}(t) < T\}, \text{ and } \\
&\tau_{-}(T) := \sup\{t \in \bbR : G_{-}(t) < T\}.
\end{align*}

Note that $\tau_\pm (T)$ is well-defined and finite, from the robustness condition. The proof is analogous to the online bidding case, as
\[
G_+(\rho T) \ge \frac{1}{\rho} \int_{-\infty}^{\rho T} G_+(t)\dd t + \frac{1}{\rho}\int_{-\infty}^{\rho T} G_-(t) \dd t \ge \frac{1}{\rho}\int_{0}^{\rho T} 1\dd t = T,
\]
implying $\tau_+(T) \le \rho T < \infty$ for any $T > 0$. Also,
\[
G_-(\rho T - 1) \ge \frac{1}{\rho} \int_{-\infty}^{\rho T} G_+(t)\dd t + \int_{-\infty}^{\rho T - 1}G_-(t) \dd t \ge \frac{1}{\rho}\int_{0}^{\rho T} \dd t = T,
\]
implying $\tau_-(T) \le \rho T - 1$ for any $T > 0$.

For a target at $+x$ $(x > 0)$, observe that if $G_{+}(n + U) < x$, cost of $2G_{+}(n + U) + 2G_{-}(n + U)$ is incurred. Otherwise, the strategy terminates after paying $x$. Therefore, the entire cost is at most
\begin{align*}
&x + 2\sum_{n \in \bbZ} \int_{0}^{1} (G_{+}(n + u) + G_{-}(n + u))\mathbf{1}_{G_+(n+u) < x}\dd u\\
&=x + 2\int_{-\infty}^{\tau_+(x)} (G_{+}(v) + G_{-}(v))\dd v\\
&= x + 2C_{+}(\tau_+(x))\\
&\le x + 2\rho G_{+}(\tau_+(x)) \le (1 + 2\rho) x.
\end{align*}
Symmetrically for a target at $-x$ $(x > 0)$, cost of $2G_{-}(n + U) + 2G_{+}(n+1+U)$ is incurred if $G_{-}(n + U) < x$. Otherwise, the strategy terminates after paying $x$. Then the entire cost is at most
\begin{align*}
&x + 2\sum_{n \in \bbZ} \int_{0}^{1} (G_{-}(n + u) + G_{+}(n + 1 + u))\mathbf{1}_{G_-(n+u) < x}\dd u\\
&=x + 2\int_{-\infty}^{\tau_-(x)} (G_{+}(v+1) + G_{-}(v))\dd v\\
&= x + 2C_{-}(\tau_-(x))\\
&\le x + 2\rho G_{-}(\tau_-(x)) \le (1 + 2\rho) x.
\end{align*}
Hence the strategy is $(1 + 2\rho)$-robust. From the plus-offset condition, it is clear that the strategy is $(1 + 2\chi)$-consistent.
\end{proof}

%% file: arxiv_include/D3_lin_ub.tex
\subsection{Deferred proofs in \Cref{subsec:lin:tight}} \label{app:lin:ub}

\lemlintight*
\begin{proof}
For a pair of nonnegative functions \(H_\pm\) on \((-\infty,0]\), let us define a functional operator $\cF$ as follows:
\[
    (\cF H)_+(x)
    :=
    \frac1\rho
    \left(
    \int_{-\infty}^{x}H_+(t)\dd t+
    \int_{-\infty}^{x}H_-(t)\dd t
    \right),
\]
and
\[
    (\cF H)_-(x)
    :=
    \frac1\rho
    \left(
    \int_{-\infty}^{\min\{x+1,0\}}H_+(t)\dd t
    +
    \mathbf 1_{\{x>-1\}}\int_0^{x+1}\psi(t)\dd t
    +
    \int_{-\infty}^{x}H_-(t)\dd t
    \right).
\]
The operator \(\cF\) is order-preserving.  Start from
\(G_+^{(0)}=G_-^{(0)}\equiv0\) on \((-\infty,0]\), and set
\(G^{(m+1)}=\cF G^{(m)}\).  Then
\[
    0=G_\pm^{(0)}\le G_\pm^{(1)}\le G_\pm^{(2)}\le\cdots .
\]
Here, $G_\pm^{(m)}\le G_\pm^{(m+1)}$ represents $G_+^{(m)}\le G_+^{(m+1)}$ and $G_-^{(m)}\le G_-^{(m+1)}$.

We next bound the iteration from above.  Choose \(c\in(s,s_\ast)\).  Since
\(r\mapsto(1+e^r)/(2r)\) is strictly decreasing on \((0,s_\ast]\), we have
\[
    \frac{1+e^c}{2c}<\rho.
\]
Equivalently, \(2c\rho-1>e^c\).  Put
\[
    \eta:=\frac{e^{2c}}{2c\rho-1}.
\]
Then \(0<\eta\le2c\rho-1\).  For a sufficiently large constant \(L\), define
\[
    W_+(x):=Le^{2cx},
    \qquad
    W_-(x):=L\eta e^{2cx},
    \qquad x\le0.
\]
We claim that \(\cF W\le W\).  For the plus coordinate,
\[
    (\cF W)_+(x)=\frac{1+\eta}{2c\rho}Le^{2cx}\le W_+(x).
\]
For the minus coordinate and \(x\le-1\),
\[
    (\cF W)_-(x)
    =
    \frac{e^{2c}+\eta}{2c\rho}Le^{2cx}
    =
    L\eta e^{2cx}
    =
    W_-(x).
\]
For \(-1<x\le0\), $(\cF W)_-(x) \le W_-(x)$ is equivalent to
\[
    L\left(
    \rho\eta e^{2cx}
    -
    \frac{1}{2c}
    -
    \frac{\eta e^{2cx}}{2c}
    \right)
    \ge
    \int_0^{x+1}\psi(t)\dd t.
\]
Let $f(x)$ denote the expression in parenthesis.
Since $f(-1)=0$ and $f'(x)=e^{2c(x+1)} \ge 1$ for $-1<x\le 0$,
\[
f(x)=f(-1)+\int_{-1}^{x}f'(t)\dd t \ge 0+\int_{-1}^{x}1\dd t = x+1,
\]
on $[-1,0]$.

Taking
\(L\ge\sup_{0<t\le1}\psi(t)\) proves the claim.

Thus \(0\le G_\pm^{(m)}\le W_\pm\) for all \(m\).  The monotone convergence theorem
gives finite limits \(G_\pm\), and the exponential bound gives finite lower-tail
integrals.  Passing to the limit in \(G^{(m+1)}=\cF G^{(m)}\) gives \(G=\cF G\),
which is exactly the pair of tight equations on \(x\le0\).  Minimality follows
from the same induction: any $(\rho, \widetilde{\chi})$-excursion profile $\widetilde{G}_{\pm}$ satisfies $\widetilde{G} \ge \cF \widetilde{G}$, giving that $\chi \le \widetilde{\chi}, G \le \widetilde{G}$.

It remains to prove \eqref{eq:lin:tight_extension_identity}.  Define
\[
    A_\pm(x):=\int_{-\infty}^{x}G_\pm(t)\dd t .
\]
The tight equations give, almost everywhere on \((-\infty,0]\),
\begin{align}
    \rho A_+'(x)=A_+(x)+A_-(x),
    \qquad
    \rho A_-'(x)=A_+(x+1)+A_-(x).
    \label{eq:lin:diff_identities}
\end{align}
The exponential bound above implies when $T \to \infty$,
\[
    e^{2sT}A_\pm(-T)\to0,
    \qquad
    \int_{-T}^{1-T}e^{-2su}A_+(u)\dd u\to0 .
\]
Multiplying the two differential identities of \eqref{eq:lin:diff_identities} by \(e^{-2sx}\), integrate over
\([-T,0]\). Define
\[
X_\pm^{T} := \int_{-T}^{0} e^{-2sx}A_{\pm}(x) \dd x, \quad Y_{+}^{T} := \int_{-T}^{0} e^{-2sx}A_{+}(x+1)\dd x.
\]
Then integrating by parts,
\begin{align}
    e^sX_+^T-X_-^T+\rho A_+(0)-\rho e^{2sT}A_+(-T)=0, \label{eq:lin:lap_identity1}
\end{align}
and
\begin{align}
    e^sX_-^T-Y_+^T+\rho A_-(0)-\rho e^{2sT}A_-(-T)=0, \label{eq:lin:lap_identity2}
\end{align}
Note that
\[
Y_{+}^T = e^{2s}\int_{1-T}^{1} e^{-2sx}A_+(x)\dd x = e^{2s}\left(X_+^T +\int_{0}^{1} e^{-2sx}A_+(x)\dd x-\int_{-T}^{1-T} e^{-2sx}A_+(x)\dd x\right).
\]
Evaluating \(e^s \eqref{eq:lin:lap_identity1} + \eqref{eq:lin:lap_identity2}\) at \(T\to\infty\), the boundary terms vanish and we get
\[
    e^{2s}\int_0^1e^{-2sx}A_+(x)\dd x
    =
    \rho\left(e^sA_+(0)+A_-(0)\right).
\]
Since \(A_+(x)=A_+(0)+\int_0^x\psi(t)\dd t\) on \([0,1]\), one integration by parts gives
\[
    \int_0^1e^{-2sx}\psi(x)\dd x
    =
    e^{-2s}A_+(1)-A_+(0)
    +
    2s\int_0^1e^{-2sx}A_+(x)\dd x .
\]
Using \(2s\rho=1+e^s\) and \(A_+(1)=A_+(0)+\int_0^1\psi(t)\dd t\), this becomes
\[
    (1+e^s)\left(A_+(0)+A_-(0)\right)
    =
    \int_0^1\left(e^{2s(1-x)}-1\right)\psi(x)\dd x.
\]
This proves \eqref{eq:lin:tight_extension_identity}.
\end{proof}

\thmlinsearchub*
Before the proof, we first show that $K(s)$ is well-defined and increases with $s$. If $s \le s_{K}$,
\begin{align}
K'(s) = \frac{e^s(e^s(4s+3) + 5e^{2s} + e^{3s} - 1)}{(1+e^s)^3} > 0,\label{eq:Kbranch:increases}
\end{align}
implies that $K(s) \le K(s_{K}) = 1$ and $K(s)$ strictly increases when $s > 0$.
\begin{remark}
    For $s_{K} < s \le s_{\ast}$, the defining equation
    \[
    (e^{s} - \xi)\ln \xi + \xi(3+e^{-s}) = e^{s}(e^{s} + 2s-1)
    \]
    has a unique solution $\xi \in [1, e^{2s}]$.
\end{remark}
\begin{proof}
    Define
    \[
    F_{s}(\xi) := (e^{s} - \xi)\ln\xi + \xi(3 + e^{-s}) - e^{s}(e^{s} + 2s - 1).
    \]
    Then
    \[
    e^{s}F_{s}(1) = 3e^{s} + 1 - e^{2s}(e^{s} + 2s-1) = (1+e^s)^2 -e^s(e^{2s}+2se^{s}-1),
    \]
    which is negative for $s > s_{K}$, considering \eqref{eq:Kbranch:increases}.

    Moreover,
    \[
    F_{s}(e^{2s}) = 2s(e^{s}- e^{2s}) + 3e^{2s} + e^{s} - e^{s}(e^{s} + 2s-1) = 2e^{s}(1 + (1-s)e^{s}).
    \]
    Since $1 + (1-s^{\ast})e^{s^{\ast}} = 0$ from the definition of $s^{\ast} = 1 + W_{0}(e^{-1})$, and $x \mapsto 1 - (1-x)e^{x}$ decreases for $x > 0$, it implies that $
    F_{s}(e^{2s}) \ge 0$.

    Now, we show that $F_{s}$ is strictly increasing.
    \[
    F_{s}'(\xi) = \frac{e^s}{\xi}-\ln \xi + 2 + e^{-s},
    \]
    and
    \[
    F_{s}''(\xi) = -\frac{1}{\xi}-\frac{e^s}{\xi^2} < 0.
    \]
    Since $F_s'$ strictly decreases, it suffices to show that $F_{s}'(e^{2s}) \ge 0$.
    \[
    F_{s}'(e^{2s}) = e^{-s} -2s + 2 + e^{-s} = 2e^{-s}(1 + (1-s)e^{s}) \ge 0,
    \]
    as we proved before.
\end{proof}
\begin{proof}[Proof of \Cref{thm:lin:Kbranch_construction}]
The endpoint \(s=s_\ast\) is the classical exponential profile, so assume
\(s<s_\ast\).  Write \(K=K(s)\) and \(M=\max\{1,K\}\).  On \(x>0\), define
\[
    G_+(x):=\max\{1,Me^{2s(x-1)}\},
\]
and
\[
    G_-(x):=Me^s e^{2s(x-1)}+(K-M)e^{-s}e^{x/\rho}.
\]
The correction term in \(G_-\) appears only when \(K<1\), and it ensures that
\(G_-(0+)=Ke^{-s}\).

Let \(\psi=G_+|_{(0,1]}\), and define \(G_\pm\) on \(x\le0\) as the minimal tight
extension of \(\psi\) from \Cref{lem:lin:tight_extension}.  The consistency of this
left extension is
\[
    \chi_s=
    \frac{1}{1+e^s}
    \int_0^1
    \left(e^{2s(1-x)}-1\right)G_+(x)\dd x.
\]
If \(K<1\), then \(G_+(x)=1\) on \((0,1]\), so this is the first expression in the
statement.  If \(K\ge1\), set \(x_K=1-\log K/(2s)\).  Splitting the integral at
\(x_K\) gives the second expression.

We next record the boundary identity
\begin{align}
    \chi_s+\int_0^1G_+(x)\dd x=\rho Ke^{-s}.
    \label{eq:lin:Kbranch_boundary_identity}
\end{align}
When \(K<1\), this is exactly the definition of \(K(s)\).  When \(K\ge1\), we use
\[
    \int_0^1G_+(x)\dd x
    =
    1+\frac{K-1-\log K}{2s},
\]
and the formula for \(\chi_s\); after multiplying by \(2s(1+e^s)\), the identity is
equivalent to
\[
    (e^s-K)\log K+K(3+e^{-s})-e^s(e^s+2s-1)=0.
\]
Thus \eqref{eq:lin:Kbranch_boundary_identity} holds in both cases.  Consequently,
the tight left extension satisfies
\[
    C_-(0)=C_+(0)+\int_0^1G_+(x)\dd x=\rho Ke^{-s}=\rho G_-(0+).
\]

We now verify robustness for \(x>0\).  The negative constraint is tight.  Indeed,
\(C_-(0)=\rho G_-(0+)\), and direct differentiation gives
\[
    C_-'(x)=G_-(x)+G_+(x+1)=\rho G_-'(x).
\]
Hence \(C_-(x)=\rho G_-(x)\) for every \(x>0\).

It remains to check \(C_+(x)\le \rho G_+(x)\).  First suppose \(K\ge1\).  Then the
correction term in \(G_-\) vanishes.  At \(x=1\),
\[
    C_+(1)=C_-(0)+\int_0^1G_-(x)\dd x
    =\rho Ke^{-s}+\int_0^1e^sK e^{2s(x-1)}\dd x
    =\rho K=\rho G_+(1).
\]
When \(G_+(x) > 1\), \(C_+'(x)=G_+(x)+G_-(x)=\rho G_+'(x)\).  Otherwise, the function \(C_+\) is nondecreasing and reaches value \(\rho\)
at the right endpoint of the plateau.  Therefore \(C_+(x)\le \rho G_+(x)\) for all
\(x>0\).

Now suppose \(K<1\).  Then \(G_+(x)=1\) on \((0,1]\) and
\(G_+(x)=e^{2s(x-1)}\) for \(x\ge1\).  Also
\[
    G_-(x)=e^s e^{2s(x-1)}+(K-1)e^{-s}e^{x/\rho}.
\]
Since \(1/\rho<2s\), this function is positive and nondecreasing on \((0,\infty)\).
Moreover,
\[
    C_+(1)
    =
    C_-(0)+\int_0^1G_-(x)\dd x
    \le
    \rho e^{-s}+\int_0^1 e^s e^{2s(x-1)}\dd x
    =
    \rho.
\]
Thus \(C_+(x)\le\rho=\rho G_+(x)\) on \(0<x\le1\).  For \(x\ge1\),
\[
    \rho G_+'(x)-C_+'(x)
    =
    \rho G_+'(x)-G_+(x)-G_-(x)
    =
    (1-K)e^{-s}e^{x/\rho}\ge0.
\]
Since the inequality holds at \(x=1\), it propagates to all \(x\ge1\).

The left half-line is tight by construction, and the positive half-line was checked above.  Hence
the robustness condition holds everywhere.  Finally, \(C_+(0)=\chi_s\).  Also
\(\chi_s\le\rho\): in the case \(K<1\) this follows from the boundary identity, while in the
case \(K\ge1\) it follows from \(K\le e^{2s}\), which gives
\(\log K+Ke^{-s}\le2s+e^s\).  Therefore \(G_+(0)=\chi_s/\rho\le1\), and tightness on
\(x\le0\) gives \(G_+(x)<1\) for \(x<0\).  Since \(G_+(x)\ge1\) for \(x>0\) by construction,
the plus-offset condition is satisfied.  Thus \(G_\pm\) is a
\((\rho,\chi_s)\)-excursion profile.
\end{proof}

\rmklinasymp*

\begin{proof}
    We only provide the asymptotics of the upper bound from \Cref{thm:lin:Kbranch_construction}. When $s < s_{K}$, \[\conslin - 1 = 2\chi(s) = \frac{e^{2s}-1-2s}{s(1+e^s)} = s + \frac{s^2}{6} + O(s^4)\]. On the other hand, \[\roblin = 1 + \frac{1 + e^s}{s} = \frac{2}{s} + 2 + \frac{s}{2} + \frac{s^2}{6} + O(s^3),\]
    implying that
    \[
    \roblin - \frac{2}{\conslin - 1} = \frac{7}{3} + O(s) = \frac{7}{3} + O(\conslin - 1).
    \]
\end{proof}

%% file: arxiv_include/D4_lin_lb.tex
\subsection{Lower bound on robustness-consistency trade-off for learning-augmented linear search} \label{sec:linsearch_lb}
In this section, we provide a lower bound on the robustness-consistency trade-off of algorithms for the linear search problem.
The proof employs Yao's lemma, which provides a lower bound of the performance for every randomized algorithm by an optimal performance of the deterministic algorithm for the selected input distribution.

Analogous to the lower bound construction for the online bidding problem from \cite{shin25},
we select our input distribution to be proportional to the distance from the origin, but scaled differently by its direction and truncated by a closed interval $[-\alpha_2 ,+\alpha_1 ]$. By choosing the scaling factor and the interval appropriately, we prove the following theorem:

\begin{theorem}\label{thm:lin:lb}
For $t\in (0,1]$, any (randomized) $\left(1+2\frac{t(t+2)^2}{-t^3+3t+4}\right)$-consistent linear search algorithm has the robustness ratio of at least
\[1+4\frac{(t^2+t+1)}{t(-t^3+3t+4)}.\]
Asymptotically, any $\conslin$-consistent $\roblin$-robust algorithm must satisfy 
\[\roblin \ge \frac{2}{\conslin-1}+\frac{3}{2}+\Theta(\conslin-1)\]
as $\conslin \to 1^{+}$.
\end{theorem}
\begin{proof}
    Given a deterministic algorithm, we encode its excursions by a sign-alternating sequence $\{y_i\}_{i=1}^\infty$.
    Recall a randomized algorithm is a probability distribution over deterministic algorithms.
    As in \Cref{app:lin:prelim}, we assume that the subsequence of the same parity has a strictly increasing property in magnitude.
    
    Having these properties of excursions in mind, we capture the \emph{state} of the algorithm at a moment by a pair $\mathbf{u}\in \mathbb{R}_{\geq 0}^{\{+,-\}}$ of non-negative numbers,
    where each entry $\mathbf{u}_{\pm}$ indicates the maximum distance traveled along the half-line of corresponding sign so far.

    The state is updated every excursion starting from
    $\mathbf{u}^{0} := (0,0)$. 
    For each $j \in \bbZ_{\ge 1}$, denote the state after $j$-th excursion as $\mathbf{u}^{j}$. If the excursion point is $+y_{j} > 0$, $\mathbf{u}^{j}=(y_{j}, \mathbf{u}_{-}^{j-1})$. Otherwise, $\mathbf{u}^{j} = (\mathbf{u}_{+}^{j}, -y_{j})$. Exceptionally, when the algorithm terminates at the target $T$ at $m$-th excursion, the state stops updating with $\mathbf{u}^m_{\sgn(T)} = T$.

    We apply Yao's lemma on robustness with consistency terms:
    \begin{align}
    &\sup_{\mu\in \mathcal{M}}\inf_{y\in \mathcal{D}}\left(\mathbb{E}_{T\sim \mu}\left[\frac{\cost_{{y}}(T)}{|T|}\right]+\lambda \cdot {\cost_{{y}}(+1)}\right)\label{eq:lin_objective}\\
    &=\sup_{\mu\in \mathcal{M}}\inf_{Y\in\mathcal{A}}\left(\mathbb{E}_{T\sim \mu}\left[\frac{\cost_{Y}(T)}{|T|}\right]+\lambda \cdot {\cost_{{Y
    }}(+1)}\right)\notag\\
    &\leq\inf_{Y\in\mathcal{A}}\sup_{\mu\in \mathcal{M}}\left(\mathbb{E}_{T\sim \mu}\left[\frac{\cost_{Y}(T)}{|T|}\right]+\lambda \cdot {\cost_{{Y}}(+1)}\right)\notag\\
    &\leq \roblin+\lambda\cdot \conslin,\label{eq:lin_robcons}
    \end{align}
    where $\lambda>0$ is any given multiplier,
    $\mathcal{M}$ is a set of all target distributions,
    $\mathcal{D}$ is a set of all deterministic algorithms
    while $\mathcal{A}$ is a set of all randomized algorithms.
    By dividing the whole inequality by $1+\lambda$, note that the inequality can be fully understood as a variant of Yao's lemma, whose distribution set is conditioned to have an atom at $+1$ with size of at least $\frac{\lambda}{1+\lambda}$.
    As a lower bound of \eqref{eq:lin_objective},
    we use a distribution $\mu$ constructed as follows:
    \begin{equation}\label{eq:hiding_strategy}
    \frac{d\mu}{dm}(T)=\frac{\mathbf{1}_{[-\alpha_2 ,0)}+c\cdot \mathbf{1}_{(0,+\alpha_1 ]}}{\frac{\alpha_2^2}{2}+\frac{c\alpha_1^2}{2}}\cdot |T|,
    \end{equation}
    where
    \begin{equation}\label{eq:parameters}
    c=\frac{2t+1}{t(t+2)},\quad\alpha_1=\frac{t^2(t+2)}{t^2+t+1},\quad\alpha_2=\frac{t(t+2)}{t^2+t+1}\end{equation}
    for $t\in(0,1]$.
    Note that $\alpha_1\in[0,1]$ and $\alpha_2>0$ if $t\in(0,1]$, since $t^2(t+2),t(t+2),t^2+t+1>0$ and
    \[t^2+t+1-t^2(t+2)=(t+1)(1-t^2)\geq 0.\]
    Also note that
    \begin{align*}
        \frac{c\alpha_1^2}{2}+\frac{\alpha_2^2}{2}
        &=\frac{t^2(t+2)^2\left(\frac{2t+1}{t(t+2)}t^2+1\right)}{2(t^2+t+1)^2}\\
        &=\frac{t^2(t+2)\left({(2t+1)}t+(t+2)\right)}{2(t^2+t+1)^2}\\
        &=\frac{t^2(t+2)}{t^2+t+1}=\alpha_1 .
    \end{align*}
    We also use the following multiplier:
    \begin{equation}
        \label{eq:multiplier}
        \lambda=\frac{1-t^2}{t^2(t+2)}.
    \end{equation}
    \eqref{eq:parameters} and \eqref{eq:multiplier} are selected appropriately to provide the best lower bound curve conditioned on the distribution $\mu$ having the form of \eqref{eq:hiding_strategy}
    
    Let any deterministic algorithm $y$ be given. Since the support of the distribution of a target and the prediction point is contained in $[-\alpha_2,+1]$,
    we can also assume that the final state of $y$ is $\mathbf{f}:=(1,\alpha_2)\in\mathbb{R}_{\geq 0}^{\{+,-\}}$.
    Regarding $y$ as a single state sequence $\{\mathbf{u}^j\}_{j=0}^m$,
    which defines a non-decreasing curve from $\mathbf{u}^0=(0,0)$ to $\mathbf{u}^m=\mathbf{f}$ whose directions should be parallel to either of two axes, 
    the cost of reaching any points $T\in[-\alpha_2 ,+1]$ can be written as follows:
    \begin{equation}\cost_{y}(T)=2\sum_{j=1}^m \mathbf{u}_{\dir(j)}^j\cdot\mathbf{1}_{\mathbf{u}_{\sgn(T)}^j<|T|}+|T|,\end{equation}
    where $\dir(j)\in \{+,-\}$ is the direction taken from $j$-th excursion. Therefore, the expected cost following the target distribution of $\mu$ is as follows.
    \begin{align}
        &\mathbb{E}_{T\sim \mu}\left[\frac{\cost_{y}(T)}{|T|}\right] \notag
        \\
        &=\frac{1}{\alpha_1 }\left(\int_0^{\alpha_1}\frac{\cost_{y}(T)}{|T|}\cdot c|T|\,\dd T+ \int_0^{\alpha_2}\frac{\cost_{y}(T)}{|T|}\cdot |T|\,\dd T\right) \notag\\
        &=1+\frac{2}{\alpha_1 }\left(c\int_0^{\alpha_1 }\sum_{j=1}^m \mathbf{u}_{\dir(j)}^j\cdot\mathbf{1}_{\mathbf{u}_+^j<T}\,\dd T+\int_0^{\alpha_2 }\sum_{j=1}^m \mathbf{u}_{\dir(j)}^j\cdot\mathbf{1}_{\mathbf{u}_-^j<T}\,\dd T\right) \notag\\
        &=1+\frac{2}{\alpha_1 }\sum_{j=1}^m\mathbf{u}_{\dir(j)}^j\cdot\left(c\int_{\min\{\mathbf{u}_+^j,\alpha_1 \}}^{\alpha_1 }\dd T+\int_{\mathbf{u}_-^j}^{\alpha_2 }\dd T\right) \notag\\
        &=1+\frac{2}{\alpha_1 }\sum_{j=1}^m\mathbf{u}_{\dir(j)}^j\cdot\left(c\cdot\max\{\alpha_1 -\mathbf{u}_+^j,0\}+(\alpha_2 -\mathbf{u}_-^j)\right).\label{eq:roblin_det}
    \end{align}
    On the other hand, the cost of reaching $+1$ is as follows.
    \begin{equation}
    \label{eq:conslin_det}
        \cost_{y}(+1)=1+2\sum_{j=1}^m \mathbf{u}_{\dir(j)}^j\cdot\mathbf{1}_{\mathbf{u}_+^j<1}.
    \end{equation}
    
    Next, we derive the conditions for $y$ that minimize the inner part of \eqref{eq:lin_objective}. Suppose $x := y_{j}$ for some $j \in \bbZ_{\ge 1}$. Note that $x \in [-\alpha_2, 0) \cup (0, +1]$. We show that $x$ is degenerated to specific points to minimize the cost. For this purpose, assume that $x \in (-\alpha_2, 0) \cup (0, \alpha_1) \cup (\alpha_1, 1)$.
    
    If $x\in(-\alpha_2,0)$, the consistency term \eqref{eq:conslin_det} is linear in $|x|$, while the robustness term \eqref{eq:roblin_det} is negative quadratic in $|x|$. Therefore, the total cost is concave in $|x|$. Hence, changing either $x \gets -\mathbf{u}_{-}^{j-2}$ or $x \gets -\mathbf{u}_{-}^{j+2}$ minimizes the cost, which eventually degenerates the step visiting $x$ in either cases. If $x\in(+\alpha_1,+1)$, then changing every $x \gets +\alpha_1$ only reduces both the robustness and consistency. Otherwise, $x\in(0,+\alpha_1 )$. Analogously the concavity of the cost with respect to $|x|$ implies that either $x \gets +\mathbf{u}_{+}^{j-2}$ or $x \gets +\mathbf{u}_{+}^{j+2}$ reduces the total cost, which eventually degenerates the step visiting $x$ in either cases.
    
    Removing degeneracies, the excursion points of the optimal deterministic algorithm must belong to $\{+1, +\alpha_1, -\alpha_2\}$. This limits the candidate sequences to 
    \begin{itemize}
        \item $0 \to +1 \to -\alpha_2$,
        \item $0 \to -\alpha_2 \to +1$,
        \item $0 \to +\alpha_1 \to -\alpha_2 \to +1$.
    \end{itemize}
    
    Combining \eqref{eq:lin_objective}, \eqref{eq:lin_robcons}, \eqref{eq:roblin_det}, and \eqref{eq:conslin_det} with these candidates,
    \begin{equation}\label{eq:lin_maximizer}
        \roblin+\lambda\cdot\conslin\geq 1+\lambda+2\alpha_2\min\left\{\frac{1}{\alpha_1},c+\lambda,1+\left(\frac{\alpha_1}{\alpha_2}+1\right)\lambda\right\}.
    \end{equation}
    Plugging \eqref{eq:parameters} and \eqref{eq:multiplier} to \eqref{eq:lin_maximizer} yields
    \begin{equation*}
        \roblin+\lambda\cdot\conslin\geq
        1+\lambda+2\min\left\{\frac{1}{t},\frac{1}{t},\frac{1}{t}\right\}=1+\lambda+\frac{2}{t},
    \end{equation*}
    thus
    \begin{equation*}
        t^2(t+2)\frac{\roblin-1}{2}+(1-t^2)\frac{\conslin-1}{2}\geq t(t+2).
    \end{equation*}
    Therefore, if $\conslin=1+2\frac{t(t+2)^2}{-t^3+3t+4}$, then
    \begin{equation*}
        t^2(t+2)\frac{\roblin-1}{2}+(1-t^2)\cdot\frac{t(t+2)^2}{-t^3+3t+4}\geq t(t+2),
    \end{equation*}
    thus
    \begin{equation*}
        \frac{\roblin-1}{2}\geq\frac{1}{t}\cdot \left(1-(1-t^2)\cdot\frac{t+2}{-t^3+3t+4}\right)
        =\frac{1}{t}\cdot\frac{2(t^2+t+1)}{-t^3+3t+4},
    \end{equation*}
    which implies
    \begin{equation*}
        \roblin\geq 1+4\frac{(t^2+t+1)}{t(-t^3+3t+4)}.
    \end{equation*}
    
    Now we verify the claimed asymptotics. We first subtract the lower bound obtained by $\frac{2}{\conslin-1}$, which is
    \begin{align*}
        1+4\frac{(t^2+t+1)}{t(-t^3+3t+4)}-\frac{-t^3+3t+4}{t(t+2)^2}&=1+\frac{-t^6+10t^4+28t^3+27t^2+8t}{t(-t^3+3t+4)(t+2)^2}\\
        &=1+\frac{-t^5+10t^3+28t^2+27t+8}{(-t^3+3t+4)(t+2)^2}.
    \end{align*}
    Since $\conslin\to 1^+$ as $t\to 0^+$, we have
    \begin{align*}
        \roblin-\frac{2}{\conslin-1}&\geq 1+\lim_{t\to 0^+}\frac{-t^5+10t^3+28t^2+27t+8}{(-t^3+3t+4)(t+2)^2}+\Theta(\conslin-1)\\
        &=\frac{3}{2}+\Theta(\conslin-1),
    \end{align*}
    which concludes the proof.
\end{proof}

%% file: arxiv_include/E_str2alg.tex
\section{Sufficiency of strategies}
\label{app:str2alg}
In this appendix, we discuss that it is sufficient to study strategies in analyzing the robustness-consistency trade-off of learning-augmented algorithms.
We may focus on online bidding below; linear search can be addressed in an analogous way.

We first argue that any bidding strategy can be converted into an algorithm for the setting where the target value $T$ is guaranteed to be at least some constant $c>0$, with the same performance guarantees. We can still assume that $P = 1$ by normalizing the universe $[c, \infty)$ by $P \geq c$ from the beginning.

Note that any reasonable deterministic strategy $x = \{x_n\}_{n \in \bbZ}$ must satisfy
\begin{equation} \label{eq:str2alg:detprop}
    \lim_{n \to -\infty} x_n = 0;
\end{equation}
otherwise, the strategy would be non-robust.
Therefore, we can always find the minimum index $n_c(x)$ at which the sequence has a value at least $c$, i.e., $n_c(x) := \min \{n \in \bbZ : x_n \geq c\}$.

Consider the bidding algorithm represented by the suffix of $x$ from index $n_c(x)$, i.e., the algorithm sequentially submits bids $x_{n_c(x)}, x_{n_c(x) + 1}, x_{n_c(x) + 2}, \ldots$ in this order until bidding at least $T$.
It is easy to see that the algorithm always incurs less cost than the strategy by the cut-off prefix sum $\sum_{n < n_c(x)} x_n$, implying that the robustness and consistency follow from the strategy.

For the randomized setting, note again that any reasonable randomized strategy satisfies \eqref{eq:str2alg:detprop} almost surely to guarantee a bounded expected robustness.
Hence, by sampling a deterministic strategy from the randomized strategy and executing the above process with this sampled strategy would yield a randomized bidding algorithm with the same robustness and consistency guarantees.

It remains to justify that the lower bound on the robustness-consistency trade-off for bidding strategies also applies to bidding algorithms.
We replicate the proof of \Cref{lem:bidding_g_convert}, but subject to a $\rho$-robust, $\chi$-consistent algorithm $\cA$. For each $P$, denote the aggregate measure of (normalized) bids of $\cA$ equipped with prediction $P$ as $\mu_{P}$. Without loss of generality, we discard the bids below $1/P$. From the robustness of $\cA$, the local bounds are reproduced in $\mu_P$ as well.
\begin{align}
\frac{1}{2} \le \mu_{P}([T, (2\rho - 1)T)) \le \rho(2\rho - 1)T, \quad T \ge \frac{1}{P}. \label{bid:algo:local_finiteness}
\end{align}
One can define the inverse `profile' $G_{P} : \bbR \to [0, \infty)$ in the same way. $G_P$ is still non-decreasing, left-continuous, and satisfies the offset and consistency condition. What lacks to be a bidding profile is that $G_P$ may be zero, and the robustness condition is restricted to its support. Denote $l_P := \sup\{x : G_P(x) = 0\}$, then $\rho$-robustness of the algorithm implies
\[
\int_{l_P}^{x+1} G_P(t)\dd t \le \rho G_P(x),\quad x >l_P.
\]
Now we take $G_P$ to the limit of $P \to \infty$. Note that $l_P$ is also unbounded. From \eqref{bid:algo:local_finiteness},
\[
-l_{P} = m((l_P, 0)) = \mu_{P}([1/P, 1)) \ge \frac{1}{2}\lfloor \log_{(2\rho - 1)}P\rfloor \to \infty,
\]
as $P \to \infty$. Therefore, define
\[
G_{\infty}(x) := \liminf_{P \to \infty} G_{P}(x).
\]
Therefore, for any $N > 0$, $l_{P} < -N$ for sufficiently large $P$. Note that $G_{\infty}$ is still non-decreasing, but not necessarily left-continuous.

From Fatou's lemma, for any infinitesimal $\varepsilon > 0$ and $x > -N + \varepsilon$,
\begin{align*}
\int_{-N}^{x+1-\varepsilon} G_{\infty}(t)\dd t &\le \liminf_{P \to \infty} \int_{-N}^{x+1-\varepsilon} G_{P}(t)\dd t \\&\le \liminf_{P \to \infty} \int_{l_P}^{x+1-\varepsilon} G_P(t) \dd t \\
&\le \liminf_{P \to \infty} \rho G_P(x-\varepsilon) = \rho G_{\infty}(x-\varepsilon).
\end{align*}
Letting $N \to \infty$ and $\varepsilon \downarrow 0$, we obtain \[
\int_{-\infty}^{x+1} G_{\infty}(t)\dd t \le \rho G_{\infty}(x-),
\]
for any $x \in \bbR$. Therefore, the left-continuous representative of $G_{\infty}$ preserves the $\rho$-robustness condition. The $\chi$-consistency condition follows analogously. Note that, from \Cref{rmk:bid:positivity}, $G_{\infty}$ is globally positive.  One subtle issue is the offset condition; $G_{\infty}(x)$ might be $1$ for some $x < 0$, but the robustness condition implies that $\tau := \sup\{x : G_{\infty}(x) < 1\}$ is nonpositive and finite. Putting \[
G(x) := G_{\infty}(\tau + x)\]
satisfies the offset condition additionally, securing all the other conditions.

Therefore, $G$ is a $(\rho, \chi)$-bidding profile, yielding a $\rho$-robust, $\chi$-consistent strategy by \Cref{cor:G_bidding}.